\documentclass[11pt]{article}
\pdfoutput=1
\usepackage[utf8]{inputenc}
\usepackage[
 a4paper,
 total={160mm,257mm},
 left=25mm,
 top=20mm]{geometry}
\usepackage[pagebackref]{hyperref}
\usepackage{amsfonts}
\usepackage{amsmath}
\usepackage{amssymb}
\usepackage{mathrsfs}  
\usepackage{xcolor}
\usepackage{authblk}
\usepackage{physics}
\usepackage{multicol}
\usepackage{tikz}
\usetikzlibrary{cd}
\usepackage{setspace}
\usepackage{amsthm}

\definecolor{c1}{RGB}{0,90,187}
\definecolor{c2}{RGB}{255,213,0}
\hypersetup{colorlinks,breaklinks,
            linkcolor = purple,
            citecolor = teal,
            urlcolor = purple,
            }

\theoremstyle{definition}
\newtheorem{theorem}{Theorem}[section]
\newtheorem{corollary}[theorem]{Corollary}
\newtheorem{definition}[theorem]{Definition}
\newtheorem{lemma}[theorem]{Lemma}
\newtheorem{proposition}[theorem]{Proposition}
\newtheorem{example}[theorem]{Example}

\theoremstyle{remark}
\newtheorem{remark}[theorem]{Remark}

\newcommand{\ee}{\mathrm{e}}
\newcommand{\ii}{\mathrm{i}}
\newcommand{\sfq}{\mathsf{q}}
\newcommand{\sfp}{\mathsf{p}}

\numberwithin{equation}{section}

\title{\textbf{Universal Cusp Scaling in Random Partitions}} 
\author[a]{\textsc{Taro Kimura}}
\author[a,b]{\textsc{Ali Zahabi}}

\affil[a]{%
Institut de Math\'ematiques de Bourgogne, Universit\'e de Bourgogne, CNRS, France}
\affil[b]{%
London Institute for Mathematical Sciences, Royal Institution, London W1S 4BS, UK}
\date{}

\begin{document}

\maketitle

\begin{abstract}
We study the universal scaling limit of random partitions obeying the Schur measure.
Extending our previous analysis~\cite{Kimura:2020sud}, we obtain the higher-order Pearcey kernel describing the multi-critical behavior in the cusp scaling limit. 
We explore the gap probability associated with the higher Pearcey kernel, and derive the coupled nonlinear differential equation and the asymptotic behavior in the large gap limit.
\end{abstract}

%\newpage
\tableofcontents
\vspace{1.5em}
\hrule

\section{Introduction and summary}

Universality of the eigenvalue statistics in the scaling limit is one of the key concepts in the study of random matrices~\cite{Mehta:2004RMT,Forrester:2010,Akemann:2011RMT}.
For example, the Tracy--Widom distribution~\cite{Tracy:1992rf}, which was originally introduced to describe the largest eigenvalue statistics of the Gaussian Unitary Ensemble (also known as edge statistics), is now applied to a large number of statistical problems in various contexts.
Typically, the eigenvalue density function $\rho(x)$ of the standard Hermitian single random matrix is obtained by solving a second order algebraic equation with respect to the auxiliary function, called the resolvent.
Hence, in the vicinity of the spectral edge $x_*$, the density function shows square root singularity, $\rho(x) \sim (x_* - x)^{\frac{1}{2}}$, where the eigenvalue statistics is well described by the Airy kernel~\cite{Forrester:1993vtx,Nagao:1993JPSJ}.
Replacing such a square root singularity with a higher order one, $\rho(x) \sim (x_* - x)^{\frac{1}{p}}$ with $p$ even, which would be realized in the fermion momenta distribution in the non-harmonic trap~\cite{LeDoussal:2018dls}, one may obtain the higher-order analog of the Airy kernel and the corresponding Tracy--Widom distribution~\cite{Claeys:2009CPAM,Akemann:2012bw,LeDoussal:2018dls,Cafasso:2019IMRN}.

A similar singular behavior is observed at the cusp point,  appearing in the collision limit of the eigenvalue supports, that we call the cusp statistics.
In this case, we shall apply the Pearcey kernel, an analog of the Airy kernel using the Pearcey integral function~\cite{Pearcey:1946PM}, to describe the eigenvalue statistics in the vicinity of the cusp singularity~\cite{Brezin:1998zz,Brezin:1998PREb}.
See also~\cite{Brezin:2016eax}.
Similar to the edge statistics, the cusp statistics is widely discussed in various contexts, including stochastic process~\cite{Tracy:2006CMP} and asymptotics of the combinatorial problems~\cite{Okounkov:2006CMP}.

%\subsection*{Summary of the results}

In this paper, we study the random partition distribution based on the Schur measure~\cite{Okounkov:2001SM}, which depends on two sets of parameters, $\mathsf{X} = (\mathsf{x}_i)_{i \in \mathbb{N}}$ and $\mathsf{Y} = (\mathsf{y}_i)_{i \in \mathbb{N}}$.
It has been known that the Schur measure random partition is a discrete determinantal point process:
The correlation function is obtained as a determinant of the associated discrete kernel $K(x,y)$ for $x,y \in \mathbb{Z} + \frac{1}{2}$.
See \eqref{eq:Schur_kernel} for the definition of the kernel.
In our previous paper~\cite{Kimura:2020sud} (see also~\cite{Betea:2020}), we showed that this kernel is asymptotic to the higher Airy kernel in the scaling limit under the condition $\mathsf{X} = \mathsf{Y}$.
In this paper, we relax this condition and consider generic parameters $\mathsf{X}$ and $\mathsf{Y}$. 
Then, we obtain the following scaling limit of the Schur measure kernel.
\begin{proposition}[Proposition~\ref{prop:kernel_scaling_lim}]
We have the following asymptotic behavior of the Schur measure kernel under the multicritical condition \eqref{eq:multicrit_cond},
 \begin{align}
  \lim_{\epsilon \to 0}
  \qty(\frac{\alpha_p}{\epsilon})^{\frac{1}{p+1}}
  K\qty(\frac{\beta}{\epsilon} + \qty(\frac{\alpha_p}{\epsilon})^{\frac{1}{p+1}} x, \frac{\beta}{\epsilon} + \qty(\frac{\alpha_p}{\epsilon})^{\frac{1}{p+1}} y) = %\qty( \frac{\alpha_p}{\epsilon} )^{\frac{1}{p+1}}
  K_{p\text{-Airy}}(x,y) 
  \, ,
 \end{align} 
where $(\alpha_p, \beta)$ are the constants defined in \eqref{eq:alpha_beta}, and $K_{p\text{-Airy}}(\cdot,\cdot)$ is the higher Airy kernel constructed with the higher Airy functions (Definition~\ref{def:higher_Airy_kernel}).
\end{proposition}
The kernel $K_{p\text{-Airy}}$ is identified with the known higher Airy kernel for even $p$, while it gives the higher-order analog of the Pearcey kernel for odd $p$.
We emphasize that the higher Pearcey kernel appearing in the odd $p$ scaling limit is realized only when we relax the condition $\mathsf{X} = \mathsf{Y}$.

As in the case of the ordinary Airy kernel, the higher Airy kernel is defined as an integral of the bilinear form of the corresponding higher Airy functions, which can be rewritten in the following form.
\begin{proposition}[Chistoffel--Darboux-type formula, Proposition~\ref{prop:CD_formula}]
The $p$-Airy kernel is written as follows,
\begin{align}
    K_{p\text{-Airy}}(x,y) 
    & = \frac{1}{x - y} \sum_{k=1}^p (-1)^k \operatorname{Ai}_p^{(p-k)}(x) \widetilde{\operatorname{Ai}}_p^{(k-1)}(y)
    \, ,
%    \label{eq:CD_formula}
\end{align}
where $\operatorname{Ai}_p^{(k)}$ and $\widetilde{\operatorname{Ai}}_p^{(k)}$ are $k$-th derivatives of the $p$-Airy functions (Definition~\ref{eq:Airy_def}).
\end{proposition}
This formula reproduces the known result for even $p$~\cite{LeDoussal:2018dls}, and generalizes the result for $p=3, 5$ presented in~\cite{Brezin:1998zz,Brezin:1998PREb}.
See Example~\ref{ex:Airy_kernels}.

In addition, we apply the saddle point approximation to study the asymptotic behavior of the higher Airy functions.
Although the saddle point analysis presented here is not rigorous, rather heuristic, it provides an efficient way to discuss the asymptotic behavior.
As an application of this analysis, we obtain the asymptotic behavior of the scaled density function, which is given by the diagonal value of the $p$-Airy kernel (Proposition~\ref{prop:density_fn_asymptotic}).

Once given such a kernel, one can formulate the probability such that no ``particle'' is found in the interval $I \subset \mathbb{R}$, which is called the gap probability, using the Fredholm determinant associated with the kernel.
We study the gap probability based on the higher Pearcey kernel, and obtain the underlying Hamiltonian system (Proposition~\ref{prop:Ham_sys}) similarly to the Fredholm determinant associated with other kernels, i.e., the Fredholm determinant behaves as the isomonodromic $\tau$-function.
We in particular consider the so-called level spacing distribution associated with the interval $I = [-s,s]$.
In this case, we obtain the coupled nonlinear differential equations (Proposition~\ref{prop:nlin_eqs}).
We also obtain the large gap asymptotics of the level spacing distribution.
\begin{proposition}[Proposition~\ref{prop:large_gap}]
 Let $F(s)$ be the Fredholm determinant defined by the $p$-Airy kernel with the interval $I = [-s,s]$, which provides the gap probability for the determinantal point process associated with the corresponding kernel.
 We have the following large gap behavior,
 \begin{align}
    F(s) \ \xrightarrow{s \to \infty} \ \exp \qty( - C_p s^{\frac{2}{p}+2} )
 \end{align}
 with a $p$-dependent positive constant $C_p$.
\end{proposition}
This behavior is consistent with the Forrester--Chen--Eriksen--Tracy conjecture~\cite{Forrester:1993vtx,Chen:1995uy}, relating the local behavior of the density function to the large gap behavior of the gap probability.

\subsection*{Organization of the paper}

The remaining part of this paper is organized as follows.
In \S\ref{sec:Schur_measure}, we show preliminary properties on the Schur measure. 
We in particular discuss that the correlation function is systematically obtained using the associated kernel.
In \S\ref{sec:scaling}, we analyze the wave function, which is a building block of the kernel, and show that it is asymptotic to the higher $p$-Airy function in the scaling limit for both even and odd $p$.
We study the asymptotic behavior of the higher Airy function in \S\ref{sec:higher_Airy_fn}, and show in \S\ref{sec:higher_Airy_kernel} that the kernel associated with the Schur measure is asymptotic to the higher Airy and Pearcey kernels.
In \S\ref{sec:gap_prob}, we explore the gap probability based on the higher Airy and Pearcey kernels.
After establishing the operator formalism in \S\ref{sec:op_formalism}, in \S\ref{sec:Fredholm_det} we study the Fredholm determinant in details, which yields the Hamiltonian system and also the Schlesinger equations.
In \S\ref{sec:level_spacing}, we consider the level spacing distribution, which is a specific case of the gap probability.
For this case, we obtain the coupled nonlinear differential equations for the auxiliary wave functions in \S\ref{sec:nlin_eqs}.
We then explore the asymptotic limit, which is called the large gap limit, of the level spacing distribution in \S\ref{sec:large_gap} and obtain the consistent result with the Forrester--Chen--Eriksen--Tracy conjecture.
In \S\ref{sec:single_H}, we analyze the Fredholm determinant with the interval $I = [s,\infty)$ (single Hamiltonian case) for the $p$-Airy kernel.

\subsection*{Acknowledgements}
We would like to thank Mattia Cafasso for communication and interest.
We are also grateful to the anonymous referee for constructive comments on the manuscript.
This work was supported by ``Investissements d'Avenir'' program, Project ISITE-BFC (No.~ANR-15-IDEX-0003), EIPHI Graduate School (No.~ANR-17-EURE-0002), and Bourgogne-Franche-Comté region.

\section{Schur measure}\label{sec:Schur_measure}

%\paragraph{Definition}

We consider the random distribution of partitions with the Schur measure defined as follows~\cite{Okounkov:2001SM}.

\begin{definition}[Schur measure]
We define the Schur measure on the set of partitions $\mathscr{Y}$,
\begin{align}
    \mu(\lambda) = \frac{1}{Z(\mathsf{X},\mathsf{Y})} \, s_\lambda(\mathsf{X}) s_\lambda(\mathsf{Y})
    \, , \qquad
 \lambda \in \mathscr{Y}
 \, ,
\end{align}
where $s_\lambda(\cdot)$ is the Schur function. 
We denote sets of the (infinitely many) parameters by $\mathsf{X} = (\mathsf{x}_i)_{i \in \mathbb{N}}$ and $\mathsf{Y} = (\mathsf{y}_i)_{i \in \mathbb{N}}$, and $Z(\mathsf{X},\mathsf{Y})$ is the normalization constant to be specified below.
\end{definition}

In this paper, we assume that all the parameters are real.
We also use another parametrization with the Miwa variables, 
\begin{align}
    t_n = \frac{1}{n} \sum_{i=1}^\infty \mathsf{x}_i^n
    \, , \qquad
    \tilde{t}_n = \frac{1}{n} \sum_{i=1}^\infty \mathsf{y}_i^n
    \, .
    \label{eq:Miwa_var}
\end{align}
These series are absolutely convergent if all $|\mathsf{x}_i|$, $|\mathsf{y}_i| < 1$.
The constant $Z(\mathsf{X},\mathsf{Y})$ is the partition function imposing the normalization condition, $\displaystyle \sum_{\lambda \in \mathscr{Y}} \mu(\lambda) = 1$, which is obtained via the Cauchy sum formula (assuming all $|\mathsf{x}_i|$, $|\mathsf{y}_i| < 1$),
\begin{align}
    Z(\mathsf{X},\mathsf{Y}) 
    = \sum_{\lambda \in \mathscr{Y}} s_\lambda(\mathsf{X}) s_\lambda(\mathsf{Y}) 
    = \prod_{1 \le i, j \le \infty} \qty( 1 - \mathsf{x}_i \mathsf{y}_j )^{-1} 
    = \exp \qty( \sum_{n = 1}^\infty n \, t_n \tilde{t}_n )
    \, .
    \label{eq:Schur_part_fn}
\end{align}

%\paragraph{Correlation function}

We then define correlation functions for the random partition as follows.
\begin{definition}[Correlation function]
The $k$-point correlation function associated with the Schur measure is defined with a set $W = (w_i)_{i = 1,\ldots,k} \subset \mathbb{Z} + \frac{1}{2}$ as follows,
\begin{align}
    \rho_k(W) 
    = \expval{ \prod_{w \in W} \delta_{w}({X}(\lambda)) }
    = \sum_{\lambda \in \mathscr{Y}} \mu(\lambda) \prod_{w \in W} \delta_{w}({X}(\lambda))
  \label{eq:rho_fn}
\end{align}
with the ``density function'' 
\begin{align}
    \delta_w (\mathscr{X}) =
    \begin{cases}
    1 & (w \in \mathscr{X}) \\ 0 & (w \not\in \mathscr{X})
    \end{cases}
\end{align}
and the boson-fermion map,
\begin{align}
    {X}(\lambda) = 
    \qty( x_i = \lambda_i - i + \frac{1}{2} )_{i \in \mathbb{N}}
    \subset \mathbb{Z} + \frac{1}{2}
    \, .
    \label{eq:Maya}
\end{align}
\end{definition}

%\subsection*{Wave functions}

In order to discuss the correlation functions of the random partitions, we introduce two distinct wave functions as follows.
\begin{definition}[Wave functions]
We define two distinct wave functions,
\begin{subequations}\label{eq:J_wf}
\begin{align}
    \mathscr{J}(z) & 
    % = \prod_{n = 1}^\infty \frac{1 - \mathsf{x}_n z}{1 - \mathsf{y}_n / z }    
    = \exp \qty( \sum_{n=1}^\infty \qty(t_n z^n - \tilde{t}_n z^{-n}) )
    = \sum_{n \in \mathbb{Z}} J(n) \, z^n
    \, , \qquad
    J(n) = \oint \frac{\dd{z}}{2 \pi \ii} \frac{\mathscr{J}(z)}{z^{n+1}} 
    \, , \\
    \widetilde{\mathscr{J}}(z) & 
    % = \prod_{n = 1}^\infty \frac{1 - \mathsf{y}_n z}{1 - \mathsf{x}_n / z }   
    = \exp \qty( \sum_{n=1}^\infty \qty(\tilde{t}_n z^n - {t}_n z^{-n}) )
    = \sum_{n \in \mathbb{Z}} \widetilde{J}(n) \, z^n
    \, , \qquad
    \widetilde{J}(n) = \oint \frac{\dd{z}}{2 \pi \ii} \frac{\widetilde{\mathscr{J}}(z)}{z^{n+1}}
    \, ,
\end{align}
\end{subequations}
which are biorthonormal,
\begin{align}
    \sum_{k \in \mathbb{Z}} J(n+k) \widetilde{J}(m+k) = \delta_{n,m}
    \, ,
    \label{eq:biorthonormal_J}
\end{align}
and related to each other as $\widetilde{\mathscr{J}}(z)^{-1} = \mathscr{J}(z^{-1})$.
\end{definition}
\begin{remark}
Imposing the relations $t_n = \pm \tilde{t}_n$ for all $n \in \mathbb{N}$, we have the following relations between the wave functions,
%\begin{subequations}
\begin{align}
    \mathscr{J}(z) 
    =
    \begin{cases}
    \widetilde{\mathscr{J}}(z)
    & (t_n = +\tilde{t}_n) \\
    \widetilde{\mathscr{J}}(z)^{-1}
    & (t_n = -\tilde{t}_n)
    \end{cases}
\end{align}
%\end{subequations}
\end{remark}

Then, one can describe the determinantal formula for the $k$-point correlation function defined in~\eqref{eq:rho_fn}.
\begin{proposition}[Determinantal formula~\cite{Okounkov:2001SM}]
The $k$-point correlation function of the random partitions with the Schur measure is given by a size $k$ determinant constructed from the kernel,
\begin{align}
    \rho_k(W) = \det_{1 \le i, j \le k} K(w_i, w_j)
    \, ,
    \label{eq:correlation_det_formula}
\end{align}
where the kernel associated with the Schur measure is given by
\begin{subequations}\label{eq:Schur_kernel}
\begin{align}
    K(r,s) & = 
    \frac{1}{(2 \pi \ii)^2}
    \oint_{|z|>|w|} \hspace{-1.5em} \dd{z} \dd{w} \,
%    \oint \frac{\dd{z}}{2 \pi \ii} \oint \frac{\dd{w}}{2 \pi \ii} \, 
    \frac{\mathscr{K}(z,w)}{z^{r + 1/2} w^{- s + 1/2}}
    \qquad \qty(r, s \in \mathbb{Z} + \frac{1}{2})
    \, , \\
    \mathscr{K}(z,w) & 
    = \frac{\mathscr{J}(z)}{\mathscr{J}(w)} \, \frac{1}{z-w}
    = \frac{\mathscr{J}(z) \widetilde{\mathscr{J}}(w^{-1})}{z - w}
    \, .
\end{align}
\end{subequations}
\end{proposition}
\begin{lemma}
From the formula~\eqref{eq:Schur_kernel}, we obtain the expression of the kernel using the wave functions,
\begin{align}
    K(r,s) = \sum_{k=1}^\infty J\qty(r+k-\frac{1}{2}) \widetilde{J}\qty(s+k-\frac{1}{2})
    \, .
    \label{eq:Schur_kernel_sum}
\end{align}
\end{lemma}
\begin{proof}
We first expand $(\mathscr{J}(z),\widetilde{\mathscr{J}}(w))$ with $(J(n),\widetilde{J}(m))$ in $\mathscr{K}(z,w)$ based on the definition~\eqref{eq:J_wf},
\begin{align}
    \mathscr{K}(z,w) = \sum_{n,m \in \mathbb{Z}} \frac{z^n w^{-m}}{z-w} J(n) \widetilde{J}(m)
    \, .
\end{align}
Then, substituting this expression and expanding the geometric series, we obtain
\begin{align}
    K(r,s) & = \frac{1}{(2 \pi \ii)^2} \oint _{|z|>|w|} \hspace{-1.5em} \dd{z} \dd{w} \, \sum_{n,m \in \mathbb{Z}} \sum_{k = 0}^\infty z^{n-r-\frac{1}{2}-k-1} w^{-m+s-\frac{1}{2}+k} J(n) \widetilde{J}(m)
    \nonumber \\
    & = \sum_{k=1}^\infty J\qty(r+k-\frac{1}{2}) \widetilde{J}\qty(s+k-\frac{1}{2})
    \, .
\end{align}
This completes the proof.
\end{proof}
\begin{remark}
From this expression~\eqref{eq:Schur_kernel_sum} together with the biorthonormal condition~\eqref{eq:biorthonormal_J}, we see that the kernel is projective,
\begin{align}
    \sum_{t \in \mathbb{Z}+\frac{1}{2}} K(r,t) K(t,s) = K(r,s)
    \, .
\end{align}
\end{remark}

\section{Scaling limit}\label{sec:scaling}

In this Section, we study the differential/difference equations for the wave functions in the scaling limit.
From this analysis, we see that the wave functions are asymptotic to the higher-order Airy functions, and the corresponding kernel is given by higher analogs of the Airy and Pearcey kernels.

\subsection{Wave functions}

From the expressions \eqref{eq:J_wf}, we see that the wave functions obey the differential/difference equations,\\
\begin{subequations}\label{eq:J_ODE}
\begin{minipage}{.48\textwidth}
\begin{align}
    \qty[ \sum_{n=1}^\infty n \qty( t_n z^n + \tilde{t}_n z^{-n} ) - z \pdv{}{z}] \mathscr{J}(z) & = 0
    \, , \\
    \qty[ \sum_{n=1}^\infty n \qty( t_n \nabla_x^n + \tilde{t}_n \nabla_x^{-n} ) - x] J(x) & = 0    
    \, ,
\end{align}
\end{minipage}
\begin{minipage}{.48\textwidth}
\begin{align}
    \qty[ \sum_{n=1}^\infty n \qty( \tilde{t}_n z^n + {t}_n z^{-n} ) - z \pdv{}{z}] \widetilde{\mathscr{J}}(z) & = 0    
    \, , \\
    \qty[ \sum_{n=1}^\infty n \qty( \tilde{t}_n \nabla_x^n + {t}_n \nabla_x^{-n} ) - x] \widetilde{J}(x) & = 0     
    \, ,
\end{align}
\end{minipage}\\[1em]
\end{subequations}
where we define the shift operator
\begin{align}
    \nabla_x = \exp \qty(\dv{}{x})
    \, , \qquad
    \nabla_x f(x) = f(x+1)
    \, .
\end{align}
We rescale the variables $(x,t_n,\tilde{t}_n) \to (x/\epsilon,t_n/\epsilon,\tilde{t}_n/\epsilon)$ to take the scaling limit.
Then, the differential/difference equations are written as follows,
%\begin{subequations}
\begin{align}
    \qty[ \sum_{k=1}^\infty \alpha_k \epsilon^k \dv[k]{}{x} - (x - \beta) ] J\qty(\frac{x}{\epsilon}) = 0
    \, , \qquad
    \qty[ \sum_{k=1}^\infty \tilde{\alpha}_k \epsilon^k \dv[k]{}{x} - (x - \beta) ] \widetilde{J}\qty(\frac{x}{\epsilon}) = 0 
\end{align}
%\end{subequations}
where we define the parameters,
\begin{align}
    \alpha_k = \sum_{n=1}^\infty \frac{n^{k+1}}{k!} \, \qty( t_n + (-1)^k \tilde{t}_n )
    \, , \qquad
    \tilde\alpha_k = (-1)^k \alpha_k
    \, , \qquad
    \beta = \alpha_0 = \sum_{n=1}^\infty n \, (t_n + \tilde{t}_n)
    \, .
    \label{eq:alpha_beta}
\end{align}

We introduce the scaling variable
\begin{align}
    x = \beta + \alpha_p^{\frac{1}{p+1}} \epsilon^{\frac{p}{p+1}} \xi
    \, .
    \label{eq:scaling_var}
\end{align}
Then, we define the scaling limit of the wave functions,
\begin{align}
    \phi(\xi) := \lim_{\epsilon \to 0 } \qty(\frac{\alpha_p}{\epsilon})^{\frac{1}{p+1}} J\qty(\frac{\beta}{\epsilon} + \qty(\frac{\alpha_p}{\epsilon})^{\frac{1}{p+1}} \xi) 
    \, , \quad
    \psi(\xi) := \lim_{\epsilon \to 0 } \qty(\frac{\alpha_p}{\epsilon})^{\frac{1}{p+1}} \widetilde{J}\qty(\frac{\beta}{\epsilon} + \qty(\frac{\alpha_p}{\epsilon})^{\frac{1}{p+1}} \xi)
    \, .
\end{align}
\begin{lemma}\label{lemma:biorthonormal_wf}
 The scaled wave functions are formally biorthonormal,
 \begin{align}
  \int_{-\infty}^{+\infty} \dd{z} \phi(x + z) \psi(y + z) = \delta(x - y)
  \, .
  \label{eq:biorthonormal_wf}
 \end{align}
\end{lemma}
\begin{proof}
 Let $f(\cdot)$ be a test function on $\mathbb{R}$.
 From the biorthonormality of $J$ and $\widetilde{J}$ shown in \eqref{eq:biorthonormal_J}, we have
 \begin{align}
  \sum_{k,k' \in \mathbb{Z}} J(n+k) \widetilde{J}(k'+k) f(k') = f(n)
  \, .
 \end{align}
 Then, this equation reads
 \begin{align}
 \sum_{k,k' \in \mathbb{Z}} J\qty(\frac{\beta}{\epsilon} + \qty(\frac{\alpha_p}{\epsilon})^{\frac{1}{p+1}} \qty(x+z) ) \widetilde{J}\qty(\frac{\beta}{\epsilon} + \qty(\frac{\alpha_p}{\epsilon})^{\frac{1}{p+1}} \qty(y+z) ) f\qty(\qty(\frac{\alpha_p}{\epsilon})^{\frac{1}{p+1}} y)
 = f\qty(\qty(\frac{\alpha_p}{\epsilon})^{\frac{1}{p+1}} x)
  \label{eq:biortho_JJf}
 \end{align}
where we have the following scaling variables,
\begin{align}
 x = \qty(\frac{\epsilon}{\alpha_p})^{\frac{1}{p+1}} n
 \, , \qquad 
 y = \qty(\frac{\epsilon}{\alpha_p})^{\frac{1}{p+1}} k'
 \, , \qquad
 z = \qty(\frac{\epsilon}{\alpha_p})^{\frac{1}{p+1}} \qty(k - \frac{\beta}{\epsilon})
 \, .
\end{align}
The summation over $\mathbb{Z}$ is converted to the Riemann integral,
\begin{align}
 \lim_{\epsilon \to 0} \qty(\frac{\epsilon}{\alpha_p})^{\frac{1}{p+1}} \sum_{k \in \mathbb{Z}} f\qty(\qty(\frac{\epsilon}{\alpha_p})^{\frac{1}{p+1}} k) = \int_{-\infty}^\infty \dd{x} f(x)
 \, .
\end{align}
Hence, taking the limit $\epsilon \to 0$ of \eqref{eq:biortho_JJf}, we obtain
\begin{align}
 \int_{-\infty}^\infty \dd{y} \dd{z} \phi(x+z) \psi(y+z) \bar{f}(y) = \bar{f}(x)
\end{align}
with the scaled test function,
 \begin{align}
  \bar{f}(x) = \lim_{\epsilon \to 0} f\qty(\qty(\frac{\alpha_p}{\epsilon})^{\frac{1}{p+1}} x)
  \, ,
 \end{align}
 from which we deduce the formula~\eqref{eq:biorthonormal_wf}.
\end{proof}

In the scaling limit, we obtain the simplified differential equations as follows. 
\begin{lemma}[Differential equations in the scaling limit]
We obtain the differential equations in the scaling limit $\epsilon \to 0$,
\begin{align}
    \qty[ \dv[p]{}{\xi} - \xi ] \phi(\xi) = 0
    \, , \qquad
    \qty[ (- 1)^p \dv[p]{}{\xi} - \xi ] \psi(\xi) = 0    
    \, ,
    \label{eq:Airy_ODE}
\end{align}
under the assumption,
\begin{align}
    \alpha_{p'} \xrightarrow{\epsilon \to 0} 0 \qquad \text{for} \qquad 0 < p' < p \, .
    \label{eq:multicrit_cond}
\end{align}
We call \eqref{eq:multicrit_cond} the multicritical condition of degree $p$.
\end{lemma}
\begin{proof}
The differential operator with respect to the $x$-variable is rewritten using the scaling variable $\xi$ as follows,
\begin{align}
    \dv{}{x} = \alpha_p^{-\frac{1}{p+1}} \epsilon^{-\frac{p}{p+1}} \dv{}{\xi} 
    \, .
\end{align}
Then, the differential equation for $J(x/\epsilon)$ is given by
\begin{align}
    0 & = 
    \qty[ \sum_{k=1}^\infty \alpha_k \alpha_p^{-\frac{k}{p+1}} \epsilon^{k-\frac{kp}{p+1}} \dv[k]{}{\xi} - \alpha_p^{\frac{1}{p+1}} \epsilon^{\frac{p}{p+1}} \xi ] J\qty(\frac{\beta}{\epsilon} + \qty(\frac{\alpha_p}{\epsilon})^{\frac{1}{p+1}} \xi) 
    \nonumber \\
    & =
    \alpha_p^{\frac{1}{p+1}} \epsilon^{\frac{p}{p+1}}
    \qty[ \sum_{k=1}^\infty \alpha_k \alpha_p^{-\frac{k+1}{p+1}} \epsilon^{k-\frac{(k+1)p}{p+1}} \dv[k]{}{\xi} - \xi ] J\qty(\frac{\beta}{\epsilon} + \qty(\frac{\alpha_p}{\epsilon})^{\frac{1}{p+1}} \xi) 
    \, .
\end{align}
Under the assumption $\alpha_{p'} \to 0$ for $p' < p$, it becomes
\begin{align}
    \alpha_p^{\frac{1}{p+1}} \epsilon^{\frac{p}{p+1}} 
    \qty[ \dv[p]{}{\xi} -  \xi + O(\epsilon^{\frac{1}{p+1}}) ] J\qty(\frac{\beta}{\epsilon} + \qty(\frac{\alpha_p}{\epsilon})^{\frac{1}{p+1}} \xi) = 0
    \, ,
\end{align}
Taking the scaling limit $\epsilon \to 0$, we obtain the differential equation for $\phi(\xi)$.
We can apply the same analysis to obtain the differential equation for the other wave function $\psi(\xi)$.
\end{proof}

Instead of the scaling variable~\eqref{eq:scaling_var}, it is also possible to use another scaling variable,%
\footnote{%
In order that all the variables and the parameters are real, we should take \eqref{eq:scaling_var} if $\alpha_p > 0$, while \eqref{eq:scaling_var2} if $\alpha_p < 0$ for odd $p$.
}
\begin{align}
    x = \beta + \tilde{\alpha}_p^{\frac{1}{p+1}} \epsilon^{\frac{p}{p+1}} \xi
    \, .
    \label{eq:scaling_var2}
\end{align}
In this case, the wave functions behave in the scaling limit as follows,
\begin{align}
    \lim_{\epsilon \to 0 } \qty(\frac{\alpha_p}{\epsilon})^{\frac{1}{p+1}} J\qty(\frac{\beta}{\epsilon} + \qty(\frac{\tilde{\alpha}_p}{\epsilon})^{\frac{1}{p+1}} \xi) = \psi(\xi)
    \, , \quad
    \lim_{\epsilon \to 0 } \qty(\frac{\alpha_p}{\epsilon})^{\frac{1}{p+1}} \widetilde{J}\qty(\frac{\beta}{\epsilon} + \qty(\frac{\tilde{\alpha}_p}{\epsilon})^{\frac{1}{p+1}} \xi) = \phi(\xi)
    \, .
\end{align}
Namely, the roles of $\phi(\xi)$ and $\psi(\xi)$ are now exchanged.
We remark that, for even $p$, $\alpha_p = \tilde{\alpha}_p$ and thus $\phi(\xi) = \psi(\xi)$.
\begin{remark}\label{rem:parity}
Parity is odd for both $\qty(\mathrm{d}/\mathrm{d}\xi)^p$ and $\xi$ in the odd $p$ case. 
Hence the differential equation preserves the parity, and thus one can choose either even or odd function as a solution.
For even $p$, on the other hand, it is not the case, so that the solution is in general asymmetric.
\end{remark}

\paragraph{Inclusion of subleading terms}

As mentioned above, we can modify the differential equation via further tuning of the parameters.
We define the parameters as follows,
\begin{align}
    \alpha_k \alpha_p^{-\frac{k+1}{p+1}} \epsilon^{k-\frac{(k+1)p}{p+1}} = \rho_k
    \, , \qquad
    \tilde{\rho}_k = (-1)^k \rho_k
    \, ,
\end{align}
with $\rho_p = 1$.
Namely, the order of the constant is given by
\begin{align}
    \alpha_k = O\qty( \epsilon^{\frac{p-k}{p+1}} )
    \, .
\end{align}
Then, we have the differential equations,
\begin{align}
    \qty[ \sum_{k = 1}^p \rho_k \dv[k]{}{\xi} - \xi ] \phi(\xi) = 0
    \, ,  \qquad
    \qty[ \sum_{k = 1}^p \tilde{\rho}_k \dv[k]{}{\xi} - \xi ] \psi(\xi) = 0
    \, .
    \label{eq:mod_ODE}
\end{align}

\subsection{Higher Airy functions}\label{sec:higher_Airy_fn}

In order to describe solutions to the differential equations appearing in the scaling limit, we introduce the higher-order Airy functions as following~\cite{Pearcey:1946PM,Brezin:1998zz,Brezin:1998PREb,Tracy:2006CMP}.

\begin{definition}[Higher-order Airy functions]\label{def:Airy_def}
We define the higher-order Airy functions as follows,
\begin{subequations}\label{eq:Airy_def}
\begin{align}
    \operatorname{Ai}_{p}(z) & = 
    \int_\gamma \frac{\dd{x}}{2 \pi \ii} \exp\qty( (-1)^{n-1} \frac{x^{p+1}}{p+1} - x z)
    & (p = 2n, 2n-1)
    \\
    \widetilde{\operatorname{Ai}}_{2n-1}(z) & = 
    \int_{\tilde{\gamma}} \frac{\dd{x}}{2 \pi \ii} \exp\qty( (-1)^{n} \frac{x^{p+1}}{p+1} - x z)
    & (p = 2n-1)
\end{align}
\end{subequations}
where $\gamma$  and $\tilde{\gamma}$ are the integral contours given by
%\begin{subequations}
\begin{align}
    \gamma: \ - \ii \infty \ \longrightarrow \ + \ii \infty
    \, , \qquad
    \tilde{\gamma} = \tilde{\gamma}_+ + \tilde{\gamma}_-
    \, ,
    \label{eq:Ai_contour}
\end{align}
with the angle 
\begin{align}
     \tilde{\gamma}_\pm : \ \pm \ii \ee^{-\ii \theta} \infty \ \longrightarrow \ 0 \ \longrightarrow \ \mp \ii \ee^{+\ii \theta} 
    \infty
    \, , \qquad
    \theta = \frac{\pi}{p+1}
    \, .
\end{align}
\end{definition}

%\end{subequations}
We may rewrite $\operatorname{Ai}_p(z)$ and $\widetilde{\operatorname{Ai}}_p(z)$ as real (improper) integrals,%
\footnote{%
For even $p = 2n$, we similarly define the higher analog of the Airy function of the second kind, which shows the same oscillation amplitude for $z \to - \infty$,
\begin{align}
    \operatorname{Bi}_{2n}(z) = \int_{0}^\infty \frac{\dd{x}}{\pi} \qty( \exp \qty( - \frac{x^{2n+1}}{2n+1} + x z ) + \sin \qty( \frac{x^{2n+1}}{2n+1} + x z ) )
    \, .
\end{align}
}
\begin{subequations}\label{eq:Airy_fn_real_form}
\begin{align}
    \operatorname{Ai}_{2n}(z) & = \int_{0}^\infty \frac{\dd{x}}{\pi} \cos \qty( \frac{x^{2n+1}}{2n+1} + x z )
    \, , \quad
    \operatorname{Ai}_{2n-1}(z) = \int_{0}^\infty \frac{\dd{x}}{\pi} \exp\qty( - \frac{x^{2n}}{2n} ) \cos \qty( x z )
    \, , \\
    \widetilde{\operatorname{Ai}}_{2n-1}(z) & = \int_{0}^\infty \frac{\dd{x}}{\pi} \exp\qty(- \frac{x^{2n}}{2n}) \Big( \ee^{-\sin \qty(\frac{\pi}{2n})xz} \cos \qty(\cos \qty(\frac{\pi}{2n}) xz + \frac{\pi}{2n}) 
    \nonumber \\
    & \hspace{12em}
    - \ee^{\sin \qty(\frac{\pi}{2n})xz} \cos \qty(\cos \qty(\frac{\pi}{2n}) xz - \frac{\pi}{2n}) \Big)
    \, .
\end{align}
\end{subequations}
We remark that $\operatorname{Ai}_{p = 2n - 1}(z)$ is an even function, while $\widetilde{\operatorname{Ai}}_{p = 2n - 1}(z)$ is an odd function.

%\paragraph{Biorthonormality}

\begin{lemma}%[Biorthonormality]
\label{lem:biorthonormal}
For $x,y \in \mathbb{R}$, the Airy functions obey the formal biorthonormal condition as follows,
\begin{align}
    \int_{-\infty}^\infty \dd{z} \operatorname{Ai}_p (x + z) \widetilde{\operatorname{Ai}}_p (y + z) = \delta(x - y)
    \, .
    \label{eq:biorthonormal}
\end{align}
\end{lemma}
\begin{proof}
Using the integral form of the delta function,
\begin{align}
    \delta (z) = \int_{-\infty}^\infty \frac{\dd{x}}{2 \pi} \, \ee^{\ii xz}
    \, ,
\end{align}
we have
\begin{align}
    \text{LHS of} \ \eqref{eq:biorthonormal}
    & =
    \int_{-\infty}^\infty \dd{z} \int_{-\infty}^\infty \frac{\dd{u}}{2 \pi} \int_{\tilde{\gamma}'} \frac{\dd{v}}{2\pi} \exp \qty( - \frac{u^{2n}}{2n} + \frac{v^{2n}}{2n} - \ii (ux + vy) - \ii z (u + v) )
    \nonumber \\
    & = \frac{1}{2\pi} \int_{-\infty}^\infty \dd{u} \int_{\tilde{\gamma}'} \dd{v} \exp \qty( - \frac{u^{2n}}{2n} + \frac{v^{2n}}{2n} - \ii (ux + vy) ) \delta(u + v)
    \nonumber \\
    & = 
    \int_{-\infty}^\infty \frac{\dd{u}}{2\pi} \ee^{- \ii u (x - y)}
    = \text{RHS}
\end{align}
where the modified contour is given by $\tilde{\gamma}' = \tilde{\gamma}_+' + \tilde{\gamma}_-'$ with $\tilde{\gamma}_\pm' : \ \mp \ee^{-\ii \theta} \infty \ \longrightarrow \ 0 \ \longrightarrow \ \pm \ee^{+\ii \theta} \infty$.
\end{proof}

\paragraph{Differential equations}

These Airy functions obey the differential equations,
\begin{subequations}
\begin{align}
    \qty[ (-1)^{n-1} \dv[2n]{}{z} - z ]\operatorname{Ai}_{2n}(z) = 0
    \, , \qquad &
    \qty[ (-1)^{n} \dv[2n]{}{z} - z ]\operatorname{Ai}_{2n}(-z) = 0
    \, ,
    \\
    \qty[ (-1)^{n} \dv[2n-1]{}{z} - z ]\operatorname{Ai}_{2n-1}(z) = 0
    \, , \qquad &
    \qty[ (-1)^{n-1} \dv[2n-1]{}{z} - z ]\widetilde{\operatorname{Ai}}_{2n-1}(z) = 0
    \, .
\end{align}
\end{subequations}
Therefore, the solutions to the differential equations~\eqref{eq:Airy_ODE} are given by
\begin{align}
    \phi(\xi) = 
    \begin{cases}
    \operatorname{Ai}_{p}((-1)^{n-1}\xi) & (p = 2n) \\
    \operatorname{Ai}_{p}(\xi) & (p = 4n - 1) \\
    \widetilde{\operatorname{Ai}}_{p}(\xi) & (p = 4n - 3)
    \end{cases}
    \qquad
    \psi(\xi) = 
    \begin{cases}
    \operatorname{Ai}_{p}((-1)^{n-1}\xi) & (p = 2n) \\
    \widetilde{\operatorname{Ai}}_{p}(\xi) & (p = 4n - 1) \\
    \operatorname{Ai}_{p}(\xi) & (p = 4n - 3) 
    \end{cases}
\end{align}
The normalizations are fixed by comparing their biorthonormality conditions (Lemma~\ref{lemma:biorthonormal_wf} and Lemma~\ref{lem:biorthonormal}).

\subsubsection{Asymptotic behavior: $\operatorname{Ai}_p$}\label{sec:Ai_asymp1}

Let us explore the asymptotic behavior of the higher Airy functions based on the saddle point approximation at $z \to \pm \infty$.%
\footnote{%
The asymptotics of the higher Airy function discussed in this part is known in the literature (See, e.g., \cite{Adler:2012ega}).
We show them here for the sake of self-contained presentation.
}
Although the saddle point analysis presented here is not rigorous, rather heuristic, it provides an efficient way to discuss the asymptotic behavior.
We may need an alternative approach, e.g., Riemann--Hilbert analysis, to confirm rigorously the results shown in this part.

In order to apply the saddle point approximation, we rewrite the $p$-Airy function for $z \in \mathbb{R}$,
\begin{align}
    \operatorname{Ai}_p(z) & = \int_\gamma \frac{\dd{x}}{2 \pi \ii} \, \ee^{W(x,z)}
    \, ,
\end{align}
where we introduce the potential function,
\begin{align}
    W(x,z) = (-1)^{n-1} \frac{x^{p+1}}{p+1} - xz
    \qquad (p = 2n, 2n - 1)
    \, .
\end{align}
The stationary equation with respect to this potential function is given by
\begin{align}
    0 = \pdv{W(x,z)}{x} = (-1)^{n-1} x^p - z 
    \, ,
\end{align}
which leads to the saddle point
\begin{align}
    x_*^p = 
    \begin{cases}
    - z & (n = 2m) \\ + z & (n = 2m - 1)
    \end{cases}.
    \label{eq:saddle_pt_x}
\end{align}

%\paragraph{$n = 2m$ ($p = 4m, 4m-1$)}
\paragraph{Positive $z$ regime}
We first consider the positive $z$ regime. 
For the case with $n = 2m$ ($p = 4m, 4m-1$), the saddle point~\eqref{eq:saddle_pt_x} is given by
\begin{align}
    x_* = \tilde\omega_k z^{\frac{1}{p}}
    \, , \qquad 
    \tilde\omega_k = \exp\qty(\frac{2k+1}{p}\pi \ii)
    \, , \qquad k \in \mathbb{Z}/p\mathbb{Z}
    \, .
\end{align}
Namely, $(\tilde{\omega}_k)_{k \in \mathbb{Z}/p\mathbb{Z}}$ is a $p$-th root of $-1$, such that $\tilde{\omega}_k^p = -1$.
For the case with $n = 2m - 1$, we instead consider the saddle point given by
\begin{align}
    x_* = \omega_k z^{\frac{1}{p}}
    \, , \qquad 
    \omega_k = \exp\qty(\frac{2k}{p}\pi \ii)
    \, , \qquad k \in \mathbb{Z}/p\mathbb{Z}
    \, ,
\end{align}
where $(\omega_k)_{k \in \mathbb{Z}/p\mathbb{Z}}$ is a $p$-th root of $+1$, such that $\omega_k^p = +1$.

We focus on the case with $n = 2m$ for the moment.
Expanding the potential function up to the quadratic order, we obtain
\begin{align}
    W(x,z) & \approx W(x_*,z) + \frac{1}{2} {\pdv[2]{W(x_*,z)}{x}} (x - x_*)^2 
    \nonumber \\
    & = - \frac{p}{p+1} \tilde{\omega}_k z^{1 + \frac{1}{p}} + \frac{p}{2} \tilde{\omega}_k^{-1} z^{1 - \frac{1}{p}} (x - x_*)^2
    \, .
\end{align}
Hence, the saddle point contribution with $\Re\qty(\tilde{\omega}_k) > 0$, which decays at $z \to + \infty$, is given by
\begin{align}
    & \exp\qty(- \frac{p}{p+1} \tilde{\omega}_k z^{1 + \frac{1}{p}}) \int_{-\ii \infty}^{+\ii\infty} \frac{\dd{x}}{2 \pi \ii} \exp\qty( \frac{p}{2} \tilde{\omega}_k^{-1} z^{1 - \frac{1}{p}} (x - x_*)^2 )
    \nonumber \\
    & = \exp\qty(- \frac{p}{p+1} \tilde{\omega}_k z^{1 + \frac{1}{p}}) z^{-\frac{1}{2}+\frac{1}{2p}} \frac{\tilde{\omega}_k^{\frac{1}{2}}}{\sqrt{2 p \pi }}
    \, .
\end{align}
Although we may deform the contour to pick up all the saddle points with a positive real part, we only focus on the most dominating factors, $x_* = \tilde{\omega}^{\pm 1}_{m-1} z^{\frac{1}{p}}$.
For the case with $n = 2m - 1$, we similarly take the saddle points $x_* = {\omega}^{\pm 1}_{m-1} z^{\frac{1}{p}}$.
Then, we obtain the asymptotic behavior of the $p$-Airy function at $z \to + \infty$,
\begin{align}
    \operatorname{Ai}_{p}(z) \xrightarrow{z \to + \infty} &
    %\sqrt{\frac{2}{p \pi}} z^{-\frac{1}{2}+\frac{1}{2p}} \sum_{k=0}^{m-1} \exp\qty(- \frac{p}{p+1} \cos \qty(\frac{2k+1}{p}\pi) z^{1 + \frac{1}{p}}) \cos \phi_k(z)
    % \nonumber \\ & \approx 
%    \sqrt{\frac{2}{p \pi}} z^{-\frac{1}{2}+\frac{1}{2p}} \exp\qty(\frac{-p}{p+1} \cos \qty(\frac{2m-1}{p}\pi) z^{1 + \frac{1}{p}}) %\cos \phi_{m-1}(z)
%    \nonumber \\
%    & \qquad \times
%    \cos \qty( \frac{p}{p+1} \sin \qty(\frac{2m-1}{p}\pi) z^{1 + \frac{1}{p}} - \frac{2m-1}{2p} \pi )
%    \nonumber \\ & = 
    \sqrt{\frac{2}{p \pi}} z^{-\frac{1}{2}+\frac{1}{2p}}
    \nonumber \\
    & \times
    \begin{cases}
    \displaystyle
    \exp\qty(\frac{-p}{p+1} \sin \qty(\frac{\pi}{p}) z^{1 + \frac{1}{p}}) \cos \qty( \frac{p}{p+1} \cos \qty(\frac{\pi}{p}) z^{1 + \frac{1}{p}} - \frac{\pi}{4} + \frac{\pi}{2p} )
    & (p = 2n) \\[1em]
    \displaystyle
    \exp\qty(\frac{-p}{p+1} \sin \qty(\frac{\pi}{2p}) z^{1 + \frac{1}{p}}) \cos \qty( \frac{p}{p+1} \cos \qty(\frac{\pi}{2p}) z^{1 + \frac{1}{p}} - \frac{\pi}{4} + \frac{\pi}{4p} )
    & (p = 2n - 1)
    \end{cases}
    \, .
    \label{eq:Ai_asymp1}
\end{align}
This expression is available both for $n = 2m$ and $n = 2m - 1$.

\paragraph{Negative $z$ regime}
We consider the negative $z$-regime $z < 0$.
In this case, the saddle point~\eqref{eq:saddle_pt_x} is given by
\begin{align}
    x_* = 
    \begin{cases}
    \omega_k |z|^{\frac{1}{p}} & (n = 2m) \\
    \tilde\omega_k |z|^{\frac{1}{p}} & (n = 2m-1) 
    \end{cases}
%    \omega_k = \exp\qty(\frac{2k}{p}\pi \ii)
    \qquad \text{for} \quad 
    k \in \mathbb{Z}/p\mathbb{Z}
    \, .
\end{align}
For $n = 2m$, expansion of the potential function up to the quadratic order is given by
\begin{align}
    W(x,z) \approx \frac{p}{p+1} \omega_k |z|^{1+\frac{1}{p}} + \frac{p}{2} \omega_k^{-1} |z|^{1-\frac{1}{p}} (x - x_*)^2
    \, .
\end{align}
In contrast to the previous case $z \to + \infty$, we shall take the saddle point with $\Im(\omega_k) \le 0$ to obtain a decaying contribution in $z \to - \infty$.
Then, we obtain the asymptotic behavior of the $p$-Airy function at $z \to - \infty$ from the saddle points $x_* = \omega_m^{\pm 1} |z|^{\frac{1}{p}}$ as follows,
\begin{align}
    \operatorname{Ai}_p(z) \xrightarrow{z \to - \infty} &
%    \sqrt{\frac{2}{p \pi}} |z|^{-\frac{1}{2}+\frac{1}{2p}} \sum_{k=m}^{2m-1} \exp\qty(\frac{p}{p+1} \cos \qty(\frac{2k}{p}\pi) |z|^{1 + \frac{1}{p}}) \cos \tilde\phi_k(|z|)
%    \nonumber \\ & \approx 
    \sqrt{\frac{2}{p \pi}} |z|^{-\frac{1}{2}+\frac{1}{2p}}
    \nonumber \\
    & \times
    \begin{cases}
    \displaystyle
    \cos \qty( \frac{p}{p+1} |z|^{1 + \frac{1}{p}} - \frac{\pi}{4} ) & (p = 2n)
    \\[1em] \displaystyle
    \exp\qty(\frac{-p}{p+1} \sin \qty( \frac{\pi}{2p} ) |z|^{1 + \frac{1}{p}}) \cos \qty( \frac{p}{p+1} \cos \qty( \frac{\pi}{2p} ) |z|^{1+\frac{1}{p}} - \frac{\pi}{4} + \frac{\pi}{4p} )
    & (p = 2n - 1)
    \end{cases}.    
\end{align}
This expression is in fact available both for $n = 2m$ and $n = 2m - 1$.
In particular for $p = 2n$, there is no exponential damping factor in the limit $z \to - \infty$.
We remark that $\operatorname{Ai}_{p = 2n - 1}(z)$ is an even function.

\subsubsection{Asymptotic behavior: $\widetilde{\operatorname{Ai}}_p$}\label{sec:Ai_asymp2}

We study the asymptotic behavior of $\widetilde{\operatorname{Ai}}_p(z)$ for $z \in \mathbb{R}$,
\begin{align}
    \widetilde{\operatorname{Ai}}_{2n-1}(z) & = \int_{\tilde\gamma} \frac{\dd{x}}{2 \pi \ii} \, \ee^{\widetilde{W}(x,z)}
    \, ,
    \qquad
    \widetilde{W}(x,z) = (-1)^n \frac{x^{2n}}{2n} - xz
    \, .
\end{align}
From the stationary equation with respect to this potential function, we obtain the saddle point as follows,
\begin{align}
    0 = {\pdv{\widetilde{W}(x,z)}{x}} = (-1)^n x^p - z
    \quad \implies \quad
    x_*^p = 
    \begin{cases}
    + z & (n = 2m) \\ - z &(n = 2m - 1)
    \end{cases}
\end{align}
For positive $z$, we obtain the following asymptotic behavior,
\begin{align}
    \widetilde{\operatorname{Ai}}_p(z) 
    \xrightarrow{z \to + \infty} 
    - \sqrt{\frac{2}{p\pi}} z^{-\frac{1}{2}+\frac{1}{2p}}
    \exp \qty( \frac{p}{p+1} \sin\qty(\frac{\pi}{2p}) z^{1+\frac{1}{p}} ) \cos \qty( \frac{p}{p+1} \cos \qty(\frac{\pi}{2p}) z^{1+\frac{1}{p}} - \frac{\pi}{4} - \frac{\pi}{4p} )
    \, .
    \label{eq:Ai_asymp2}
\end{align}
Since $\widetilde{\operatorname{Ai}}_p(z)$ is an odd function, $\widetilde{\operatorname{Ai}}_p(-z) = - \widetilde{\operatorname{Ai}}_p(z)$, the asymptotic behavior for the negative $z$ is similarly obtained.
We remark that the function $\widetilde{\operatorname{Ai}}_p(z)$ is not bounded at infinity, $\widetilde{\operatorname{Ai}}_p(z) \xrightarrow{z \to \pm \infty} \mp \infty$, while the product behaves as $\operatorname{Ai}_p(z) \widetilde{\operatorname{Ai}}_p(z) \xrightarrow{z \to \pm \infty} 0$.
See \S\ref{sec:density_fn} for details.

\subsection{Higher Airy kernel}\label{sec:higher_Airy_kernel}

We now discuss the asymptotic form of the kernel in the scaling limit.
\begin{definition}[Higher Airy kernel]\label{def:higher_Airy_kernel}
We define the $p$-Airy kernel with the higher Airy functions,
\begin{align}
    K_{p\text{-Airy}}(x,y) 
    & = 
    \int_{0}^\infty \dd{z}
    \operatorname{Ai}_p (x + z) \widetilde{\operatorname{Ai}}_p (y + z)
    \, .
    \label{eq:Airy_kernel_int}
\end{align}
\end{definition}
This integral is understood as an improper integral for odd $p$, while the integral converges absolutely for even $p$:
From the asymptotic behavior of the Airy functions discussed in \S\ref{sec:Ai_asymp1} and \S\ref{sec:Ai_asymp2}, we see that $\operatorname{Ai}_p(z) {\operatorname{Ai}}_p(z) \xrightarrow{z \to + \infty} z^{-1+\frac{1}{p}} \times $(exponential damping factor) for even $p$, $\operatorname{Ai}_p(z) \widetilde{\operatorname{Ai}}_p(z) \xrightarrow{z \to \pm \infty} |z|^{-1+\frac{1}{p}} \times $(oscillatory factor) for odd $p$.
See \eqref{eq:rho_derivative_asymptotic} in \S\ref{sec:density_fn} for more precise expressions.
Meanwhile, we will show an alternative expression of the $p$-Airy kernel in terms of the derivatives of the $p$-Airy functions (Proposition~\ref{prop:CD_formula}), which are clearly finite.

\if0
This integral is well-defined because the product of the $p$-Airy functions is suppressed in the asymptotic regime,
\begin{subequations}
\begin{align}
    p = 2n \ : \quad &
    \operatorname{Ai}_p(z) {\operatorname{Ai}}_p(z) \xrightarrow{z \to + \infty} 0 
    \, , \\
    p = 2n - 1 \ : \quad &
    \operatorname{Ai}_p(z) \widetilde{\operatorname{Ai}}_p(z) \xrightarrow{z \to \pm \infty} 0
    \, .
\end{align}
\end{subequations}
\fi

In the following, we call the kernel for even $p$ the higher Airy kernel, and the higher Pearcey kernel for odd $p$.
\begin{lemma}[Projectivity]
The $p$-Airy kernel is projective,
\begin{align}
    \int_{-\infty}^\infty \dd{z} K_{p\text{-Airy}}(x,z) K_{p\text{-Airy}}(z,y) = K_{p\text{-Airy}}(x,y)
    \, .
    \label{eq:projectivity_Airy}
\end{align}
\end{lemma}
\begin{proof}
This immediately follows from the biorthogonaliry of the $p$-Airy functions (Lemma~\ref{lem:biorthonormal}).
\end{proof}

Then, we obtain the following behavior of the kernel in the scaling limit.
\begin{proposition}\label{prop:kernel_scaling_lim}
The Schur measure kernel is asymptotic to the higher Airy kernel in the scaling limit under the multicritical condition \eqref{eq:multicrit_cond},
\begin{align}
 \lim_{\epsilon \to 0}
 \qty(\frac{\alpha_p}{\epsilon})^{\frac{1}{p+1}}
 K\qty(\frac{\beta}{\epsilon} + \qty(\frac{\alpha_p}{\epsilon})^{\frac{1}{p+1}} x, \frac{\beta}{\epsilon} + \qty(\frac{\alpha_p}{\epsilon})^{\frac{1}{p+1}} y) = %\qty( \frac{\alpha_p}{\epsilon} )^{\frac{1}{p+1}}
 K_{p\text{-Airy}}(x,y) 
    \, .
    \label{eq:kernel_scaling_lim}
\end{align}
\end{proposition}
\begin{proof}
The derivation is parallel with our previous paper \cite{Kimura:2020sud}.
Before taking the scaling limit, the kernel is given by the summation formula~\eqref{eq:Schur_kernel_sum}.
In terms of the scaling variables, we have
\begin{align}
    &
    \qty(\frac{\alpha_p}{\epsilon})^{\frac{1}{p+1}} K\qty(\frac{\beta}{\epsilon} + \qty(\frac{\alpha_p}{\epsilon})^{\frac{1}{p+1}} x, \frac{\beta}{\epsilon} + \qty(\frac{\alpha_p}{\epsilon})^{\frac{1}{p+1}} y)
    \nonumber \\ 
    & = \qty(\frac{\epsilon}{\alpha_p})^{\frac{1}{p+1}} \sum_{k=1}^\infty \qty(\frac{\alpha_p}{\epsilon})^{\frac{1}{p+1}} J\qty(\frac{\beta}{\epsilon} + \qty(\frac{\alpha_p}{\epsilon})^{\frac{1}{p+1}} \qty(x + \qty(\frac{\epsilon}{\alpha_p})^{\frac{1}{p+1}}(k-\frac{1}{2}))) 
    \nonumber \\ 
    & \hspace{8em} \times 
    \qty(\frac{\alpha_p}{\epsilon})^{\frac{1}{p+1}} \widetilde{J}\qty(\frac{\beta}{\epsilon} + \qty(\frac{\alpha_p}{\epsilon})^{\frac{1}{p+1}} \qty(y + \qty(\frac{\epsilon}{\alpha_p})^{\frac{1}{p+1}}(k-\frac{1}{2})))
    \, .
\end{align}
In the scaling limit $\epsilon \to 0$, we replace the wave function with the Airy functions and the Riemann sum with the integral, which provides the integral form of the higher Airy kernel~\eqref{eq:Airy_kernel_int}.
\end{proof}

As mentioned in Remark~\ref{rem:parity}, for odd $p$, one can make the wave functions either even of odd function.
Therefore, we correspondingly consider the symmetrized kernel as follows,
\begin{align}
    p = 2n - 1 \ : \
    K_{p\text{-Airy}}(x,y) =
    \frac{1}{2} \int_{0}^\infty \dd{z} \qty[
    \operatorname{Ai}_p (x + z) \widetilde{\operatorname{Ai}}_p (y + z)
    - \operatorname{Ai}_p (x - z) \widetilde{\operatorname{Ai}}_p (y - z)
    ]
    \, ,
    \label{eq:Airy_kernel_int2}
\end{align}
so that, for $p = 2n - 1$, we have
\begin{align}
    K_{p\text{-Airy}}(x,y) = K_{p\text{-Airy}}(-x,-y)
    \, .
    \label{eq:kernel_symmetry}
\end{align}

\subsubsection{Christoffel--Darboux-type formula}

We consider the Christoffel--Darboux-type formula for the higher Airy kernel.
In order to derive this formula, we start with the following property.
\begin{lemma}
For odd $p = 2n - 1$, the $p$-Airy kernel \eqref{eq:Airy_kernel_int2} obeys the following,
\begin{align}
    x K_{p\text{-Airy}}(x,0) & = \sum_{m=1}^{\lfloor p/2 \rfloor} \operatorname{Ai}_p^{(p-2m)}(x) \widetilde{\operatorname{Ai}}_p^{(2m-1)}(0)
    \, .
    \label{eq:xK(x,0)}
\end{align}
\end{lemma}
\begin{proof}
Denoting
\begin{align}
    \mathscr{A}(z) = \frac{1}{2} \qty( \operatorname{Ai}_p(z+x) - \operatorname{Ai}_p(z-x) )
    \, , \qquad
    \mathscr{A}^{(k)}(z) = \dv[k]{}{z} \mathscr{A}(z)
    \, ,
\end{align}
we have
\begin{align}
    x \mathscr{A}(z) & = \frac{1}{2} \qty( (z+x) \operatorname{Ai}_p(z+x) - (z-x) \operatorname{Ai}_p(z-x) ) - \frac{z}{2} \qty( \operatorname{Ai}_p(z+x) - \operatorname{Ai}_p(z-x) )
    \nonumber \\
    & = \mathscr{A}^{(p)}(z) - z \mathscr{A}(z)
    \, .
\end{align}
Then, we obtain
\begin{align}
    x K_{p\text{-Airy}}(x,0) 
    & =  \int_0^\infty \dd{z} x \mathscr{A}(z) \widetilde{\operatorname{Ai}}_p(z)
    \nonumber \\
    & = \int_0^\infty \dd{z} \qty[ \mathscr{A}^{(p)}(z) \widetilde{\operatorname{Ai}}_p(z) + \mathscr{A}(z) \widetilde{\operatorname{Ai}}_p^{(p)}(z)]
    \, .
\end{align}
Applying the integration by parts recursively, we have
\begin{align}
    x K_{p\text{-Airy}}(x,0) 
    & = 
    \sum_{k=1}^p (-1)^{k+1} \qty[\mathscr{A}^{(p-k)}(z) \widetilde{\operatorname{Ai}}_p^{(k-1)}(z)]_0^\infty
    \, .
\end{align}
We remark $\mathscr{A}^{(k)}(z) \xrightarrow{z \to \infty} 0$.
In addition, $\operatorname{Ai}_p^{(2k)}(x)$ and $\widetilde{\operatorname{Ai}}_p^{(2k-1)}(x)$ are even, $\operatorname{Ai}_p^{(2k-1)}(x)$ and $\widetilde{\operatorname{Ai}}_p^{(2k)}(x)$ are odd, so that $\widetilde{\operatorname{Ai}}_p^{(2k)}(0) = 0$.
Hence, the odd $k$ terms are zero in the summation, which proves the result \eqref{eq:xK(x,0)}.
\end{proof}

Then, we arrive at the following formula.
\begin{proposition}[Christoffel--Darboux-type formula]\label{prop:CD_formula}
The $p$-Airy kernel defined in~\eqref{eq:Airy_kernel_int} is written in terms of the higher Airy functions,
\begin{align}
    K_{p\text{-Airy}}(x,y) 
    & = \frac{1}{x - y} \sum_{k=1}^p (-1)^k \operatorname{Ai}_p^{(p-k)}(x) \widetilde{\operatorname{Ai}}_p^{(k-1)}(y)
    \, .
    \label{eq:CD_formula}
\end{align}
\end{proposition}
\begin{proof}
The even case $p = 2n$ is already known~\cite{LeDoussal:2018dls}.
Hence, we focus on the odd case ($p = 2n - 1$) here.
From the definition, we have the relation
\begin{align}
    \dv{}{z} K_{p\text{-Airy}}(x+z,y+z)
    = - \operatorname{Ai}_p^{}(x+z) \widetilde{\operatorname{Ai}}_p^{}(y+z)
    \, .
    \label{eq:kernel_z_derivative}
\end{align}
This is also obtained from the RHS of \eqref{eq:CD_formula}.
Therefore, we have the following expression with the constant of integration,
\begin{align}
    K_{p\text{-Airy}}(x,y)
    = \frac{1}{x - y} \sum_{k=1}^p (-1)^k \operatorname{Ai}_p^{(p-k)}(x) \widetilde{\operatorname{Ai}}_p^{(k-1)}(y) + C(x,y) \, .
\end{align}
To be compatible with another relation, 
\begin{align}
    \qty( \pdv{}{x} + \pdv{}{y} ) K_{p\text{-Airy}}(x,y) = - \operatorname{Ai}_p^{}(x) \widetilde{\operatorname{Ai}}_p^{}(y) 
    \, ,
\end{align}
we have the constraint for the constant $(\partial_x + \partial_y)C(x,y) = 0 \implies C(x,y) = C(x-y)$.
Then, putting $y = 0$ and from the relation \eqref{eq:xK(x,0)}, we conclude that the constant is zero, $C(x) = 0$. 
This proves the equality \eqref{eq:CD_formula}.
\if0
\begin{align}
    \dv{}{z} K_{p\text{-Airy}}(x+z,y+z)
    & = \frac{1}{x - y} \sum_{k=1}^p (-1)^k \qty( \operatorname{Ai}_p^{(p-k+1)}(x+z) \widetilde{\operatorname{Ai}}_p^{(k-1)}(y+z) + \operatorname{Ai}_p^{(p-k)}(x+z) \widetilde{\operatorname{Ai}}_p^{(k)}(y+z) ) 
    \nonumber \\
    & = \frac{1}{x - y} \qty( - \operatorname{Ai}_p^{(p)}(x+z) \widetilde{\operatorname{Ai}}_p^{}(y+z) + (-1)^p \operatorname{Ai}_p^{}(x+z) \widetilde{\operatorname{Ai}}_p^{(p)}(y+z))
    \nonumber \\
    & = - \operatorname{Ai}_p^{}(x+z) \widetilde{\operatorname{Ai}}_p^{}(y+z)
    \, .
\end{align}
\fi
\end{proof}

\begin{example}\label{ex:Airy_kernels}
The lower degree examples of the $p$-Airy kernel are given as follows, 
\begin{subequations}
\begin{align}
    K_{2\text{-Airy}}(x,y) & = \frac{- \operatorname{Ai}_2'(x) \operatorname{Ai}_2(y) + \operatorname{Ai}_2(x) \operatorname{Ai}'_2(y)}{x - y}
    \, , \\
    K_{3\text{-Airy}}(x,y) & = \frac{-\operatorname{Ai}_3''(x) \widetilde{\operatorname{Ai}}_3(y) + {\operatorname{Ai}}_3'(x) \widetilde{\operatorname{Ai}}_3'(y) - {\operatorname{Ai}}_3(x) \widetilde{\operatorname{Ai}}_3''(y)}{x-y}
    \, , \\
    K_{4\text{-Airy}}(x,y) & = \frac{- \operatorname{Ai}_4'''(x) \operatorname{Ai}_4(y) + \operatorname{Ai}_4''(x) \operatorname{Ai}_4'(y) - \operatorname{Ai}_4'(x) \operatorname{Ai}_4''(y) + \operatorname{Ai}_4(x) \operatorname{Ai}_4'''(y)}{x - y}
    \, , \\
    K_{5\text{-Airy}}(x,y) & = \frac{-\operatorname{Ai}_5''''(x) \widetilde{\operatorname{Ai}}_5(y) + {\operatorname{Ai}}_5'''(x) \widetilde{\operatorname{Ai}}_5'(y) - {\operatorname{Ai}}_5''(x) \widetilde{\operatorname{Ai}}_5''(y) + \operatorname{Ai}_5'(x) \widetilde{\operatorname{Ai}}_5'''(y) - \operatorname{Ai}_5(x) \widetilde{\operatorname{Ai}}_5''''(y)}{x-y}
    \, .
\end{align}
\end{subequations}
The kernel for $p = 2$ is the standard Airy kernel~\cite{Bowick:1991ky,Forrester:1993vtx,Nagao:1993JPSJ}, and the case with $p = 3$ is called the Pearcey kernel.
The case $p = 5$ has been also mentioned in~\cite{Brezin:1998zz,Brezin:1998PREb}.
\end{example}

\begin{corollary}[Diagonal value of the kernel]
The diagonal value of the kernel is given as follows,
\begin{align}
    K_{p\text{-Airy}}(x,x) 
%    & = \int_0^\infty \dd{z} \operatorname{Ai}_p (x + z) \widetilde{\operatorname{Ai}}_p (x + z)
%    \nonumber \\
    & = - x \operatorname{Ai}_p(x) \widetilde{\operatorname{Ai}}_p(x) - \sum_{k=1}^{p-1} (-1)^k \operatorname{Ai}_p^{(p-k)}(x) \widetilde{\operatorname{Ai}}_p^{(k)}(x)
%    & = \sum_{k=1}^p (-1)^k \operatorname{Ai}_p^{(p-k+1)}(x) \widetilde{\operatorname{Ai}}_p^{(k-1)}(x)
%    = \sum_{k=1}^p (-1)^{k-1} \operatorname{Ai}_p^{(p-k)}(x) \widetilde{\operatorname{Ai}}_p^{(k)}(x)
    \, ,
    \label{eq:kernel_diag}    
\end{align}
and obeys the relation
\begin{align}
    \dv{}{x} K_{p\text{-Airy}}(x,x) = - \operatorname{Ai}_p(x) \widetilde{\operatorname{Ai}}_p(x) \, .
    \label{eq:kernel_diagonal_der}
\end{align}
\end{corollary}
This relation will be useful to study the asymptotic behavior of the density function in the scaling limit.
See \S\ref{sec:density_fn}.

In order that the relation~\eqref{eq:kernel_diag} holds, we shall have the following relation.
\begin{lemma}
The following relation holds,
\begin{align}
    \sum_{k=1}^p (-1)^k \operatorname{Ai}_p^{(p-k)}(x) \widetilde{\operatorname{Ai}}_p^{(k-1)}(x) = 0  
    \, .
    \label{eq:diagonal_id}
\end{align}
\end{lemma}
\begin{proof}
The derivative of the LHS is zero,
\begin{align}
    \sum_{k=1}^p (-1)^k \dv{}{x} \qty( \operatorname{Ai}_p^{(p-k)}(x) \widetilde{\operatorname{Ai}}_p^{(k-1)}(x) )
    = 0
    \, .
\end{align}
The constant of integration can be fixed by $\operatorname{Ai}_p^{(p-k)}(x) \widetilde{\operatorname{Ai}}_p^{(k-1)}(x) = \operatorname{Ai}_p^{(p-k)}(x) {\operatorname{Ai}}_p^{(k-1)}(x) \xrightarrow{x \to \infty} 0$ for even $p$ and $\operatorname{Ai}_p^{(2k-1)}(0) = \widetilde{\operatorname{Ai}}_p^{(2k)}(0) = 0$ for odd $p$.
\end{proof}

\subsubsection{Density function}\label{sec:density_fn}

Due to the determinantal formula of the correlation function~\eqref{eq:correlation_det_formula}, the one-point function, which is called the density function, is given by the diagonal value of the kernel.
Hence, the scaled density function is obtained from the $p$-Airy kernel~\eqref{eq:kernel_diag}, $\rho(x) \to \bar{\rho}(x) = K_{p\text{-Airy}}(x,x)$, as in~\eqref{eq:kernel_scaling_lim}.
We have the following asymptotic behavior of the density function.
\begin{proposition}[Asymptotic behavior of the density function]\label{prop:density_fn_asymptotic}
Define the density function from the $p$-Airy kernel,
\begin{align}
    \bar{\rho}(x) = K_{p\text{-Airy}}(x,x)
    \, .
\end{align}
The asymptotic behaviors of the density function are given as follows:
\begin{subequations}\label{eq:density_asympt}
\begin{align}
    p = 2n \ : \
    \bar{\rho}(x) \xrightarrow{x \to + \infty} \ &
    \frac{1}{2 p \pi x} \exp\qty(\frac{-2p}{p+1} \sin \qty(\frac{\pi}{p}) x^{1 + \frac{1}{p}}) \qty[ \frac{1}{\sin \qty(\pi/p)} + \cos\qty(\frac{2p}{p+1} \cos \qty(\frac{\pi}{p}) x^{1 + \frac{1}{p}})  ]
    \\
    \xrightarrow{x \to - \infty} \ & \frac{|x|^{\frac{1}{p}}}{\pi} - \frac{\cos \qty( \frac{2p}{p+1} |x|^{1 + \frac{1}{p}} )}{2 p \pi |x|}
    \\
    p = 2n-1 \ : \
    \bar{\rho}(x) \xrightarrow{x \to \pm \infty} \ & \frac{\cos\qty(\frac{\pi}{2p})}{\pi} |x|^{\frac{1}{p}} - \frac{\cos\qty( \frac{2p}{p+1} \cos\qty(\frac{\pi}{2p}) |x|^{1 + \frac{1}{p}})}{2 p \pi \cos\qty(\frac{\pi}{2p}) |x|}
\end{align}
\end{subequations}
\end{proposition}
\begin{proof}
From the relation~\eqref{eq:kernel_diagonal_der}, we obtain
\begin{align}
    \dv{\bar{\rho}(x)}{x} = - \operatorname{Ai}_p(x) \widetilde{\operatorname{Ai}}_p(x)
    \, .
\end{align}
Then, applying the asymptotic behavior of the Airy functions obtained in \S\ref{sec:Ai_asymp1} and \S\ref{sec:Ai_asymp2}, we obtain the asymptotic behavior of the density function.
\paragraph{Even case: $p = 2n$}
In this case, the derivative of the density function is given as follows,
\begin{subequations}\label{eq:rho_derivative_asymptotic}
    \begin{align}
        \dv{\bar{\rho}(x)}{x} \xrightarrow{x \to + \infty} & 
        - \frac{2}{p \pi} x^{-1+\frac{1}{p}}
        \exp\qty(\frac{-2p}{p+1} \sin \qty(\frac{\pi}{p}) x^{1 + \frac{1}{p}}) \cos^2 \qty( \frac{p}{p+1} \cos \qty(\frac{\pi}{p}) x^{1 + \frac{1}{p}} - \frac{\pi}{4} + \frac{\pi}{2p} )
        \nonumber \\ &
        = - \frac{1}{p \pi} x^{-1+\frac{1}{p}}
        \exp\qty(\frac{-2p}{p+1} \sin \qty(\frac{\pi}{p}) x^{1 + \frac{1}{p}}) \qty[1 + \sin \qty( \frac{2p}{p+1} \cos \qty(\frac{\pi}{p}) x^{1 + \frac{1}{p}} + \frac{\pi}{p} )]
        \\
        \xrightarrow{x \to - \infty} & 
    - \frac{2}{p \pi} |x|^{-1+\frac{1}{p}}
    \cos^2 \qty( \frac{p}{p+1} |x|^{1 + \frac{1}{p}} - \frac{\pi}{4} )
    \nonumber \\ &
    = - \frac{1}{p \pi} |x|^{-1+\frac{1}{p}} \qty[ 1 +
    \sin \qty( \frac{2p}{p+1} |x|^{1 + \frac{1}{p}} ) ]
    \end{align}

\paragraph{Odd case: $p = 2n-1$}
For the odd case, we have the following asymptotic behavior,
\begin{align}
%    - \operatorname{Ai}_{p = 2n-1}(x) \widetilde{\operatorname{Ai}}_{p = 2n-1}(x) 
    \dv{\bar{\rho}(x)}{x}
    \xrightarrow{x \to \pm \infty} 
    - \frac{1}{p \pi} |x|
    ^{-1+\frac{1}{p}} \qty( \sin\qty( \frac{2p}{p+1} \cos \qty(\frac{\pi}{2p}) |x|^{1 + \frac{1}{p}} ) + \cos\qty(\frac{\pi}{2p}) )
    \, .
\end{align}
\end{subequations}
Integrating these expressions, we obtain~\eqref{eq:density_asympt}.
\end{proof}

\begin{remark}
This expression is compatible with the results for even $p$~\cite{LeDoussal:2018dls}.
\end{remark}
\begin{remark}
In the odd case $p = 2n-1$, the density function is an even function, while it is asymmetric for the even case $p = 2n$.
Asymptotic values of the density function are non-zero except for $\bar{\rho}(x)$ at $x \to +\infty$ for $p = 2n$.
\end{remark}

\subsubsection{Modification of the kernel}

As shown in \eqref{eq:mod_ODE}, we can modify the differential equations to contain the subleading terms.
In this case, the corresponding modified Airy functions are given as follows,
\begin{align}
    \operatorname{Ai}_p(z) = \int_\gamma \frac{\dd{x}}{2 \pi \ii} \ee^{W(x,z)}
    \, , \qquad
    \widetilde{\operatorname{Ai}}_p(z) = \int_{\tilde{\gamma}} \frac{\dd{x}}{2 \pi \ii} \ee^{\widetilde{W}(x,z)}
    \, ,
\end{align}
where the potential functions obey
\begin{align}
    \pdv{}{x}W(x,z) = \sum_{k=1}^p \rho_k x^k - z
    \, , \qquad
    \pdv{}{x}\widetilde{W}(x,z) = \sum_{k=1}^p \tilde\rho_k x^k - z
\end{align}
for the coefficients $(\rho_k,\tilde{\rho}_k)_{k=1,\ldots,p}$ obtained in \eqref{eq:mod_ODE}.
We remark that we apply the same integration contours $(\gamma,\tilde{\gamma})$ as before \eqref{eq:Ai_contour}.

Under this modification, the kernel is accordingly modified as follows,
\begin{align}
    K_{p\text{-Airy}}(x,y) = \frac{1}{x - y} \sum_{\ell = 1}^p \sum_{k = 1}^{\ell} (-1)^k \rho_\ell \operatorname{Ai}_p^{(\ell - k)}(x) \widetilde{\operatorname{Ai}}_p^{(k-1)}(y)
    \, ,
\end{align}
which obeys the same relation as \eqref{eq:kernel_z_derivative},
\begin{align}
    & 
    \dv{}{z} K_{p\text{-Airy}}(x+z,y+z) 
    \nonumber \\
    & = \frac{1}{x-y} \sum_{\ell = 1}^p \sum_{k = 1}^{\ell} (-1)^k \rho_\ell \qty( \operatorname{Ai}_p^{(\ell + 1 - k)}(x+z) \widetilde{\operatorname{Ai}}_p^{(k-1)}(y+z) + \operatorname{Ai}_p^{(\ell - k)}(x+z) \widetilde{\operatorname{Ai}}_p^{(k)}(y+z) )
    \nonumber \\
    & = - \frac{1}{x-y} \sum_{\ell = 1}^p \qty( \rho_\ell \operatorname{Ai}_p^{(\ell)}(x+z) \widetilde{\operatorname{Ai}}_p(y+z) + \tilde{\rho}_\ell \operatorname{Ai}_p(x+z) \widetilde{\operatorname{Ai}}_p^{(\ell)}(y+z) )
    \nonumber \\
    & = - \operatorname{Ai}_p^{}(x+z) \widetilde{\operatorname{Ai}}_p^{}(y+z)
    \, .
\end{align}
The diagonal value of the kernel is given by
\begin{align}
    K_{p\text{-Airy}}(x,x) 
    & = \sum_{\ell = 1}^p \sum_{k = 1}^{\ell} (-1)^k \rho_\ell \operatorname{Ai}_p^{(\ell + 1 - k)}(x) \widetilde{\operatorname{Ai}}_p^{(k-1)}(x)
    \nonumber \\
    & = \sum_{\ell = 1}^p \sum_{k = 1}^{\ell} (-1)^{k+1} \rho_\ell \operatorname{Ai}_p^{(\ell - k)}(x) \widetilde{\operatorname{Ai}}_p^{(k)}(x)
    \, ,
\end{align}
and from the regularity of the kernel at $x = y$, we have the relation,
\begin{align}
    \sum_{\ell = 1}^n \sum_{k = 1}^{\ell} (-1)^k \rho_\ell \operatorname{Ai}_p^{(\ell - k)}(x) \widetilde{\operatorname{Ai}}_p^{(k-1)}(x) = 0 
    \, .
\end{align}

\section{Gap probability}\label{sec:gap_prob}

Once given a kernel, the probability such that no ``particle'' is found in the interval $I \subset \mathbb{R}$, called the gap probability, is given by the Fredholm determinant,
\begin{align}
    \det (1 - \hat{K}) = \sum_{k=0}^\infty \frac{(-1)^k}{k!} \int_I \dd{x}_1 \cdots \int_I \dd{x}_k \det_{1 \le i, j \le k} K(x_i, x_j)
    \, ,
\end{align}
where $\hat{K}$ is the integral operator defined in~\eqref{eq:K_op_I}.
In this Section, we explore the Fredholm determinant associated with the higher Airy kernel for both $p$ even and odd.
Hence, it provides the gap probabilities of the determinantal point process defined by the limiting kernel.
We will consider the specific cases $I = [-s,s]$  in \S\ref{sec:level_spacing} and $I = [s,\infty)$ in \S\ref{sec:single_H}.

\subsection{Operator formalism}\label{sec:op_formalism}

\paragraph{Definitions}

We introduce the operator formalism to study the Fredholm determinant.
For any integrable function $f(x)$, we define the corresponding state $\ket{f}$ in the Hilbert space,
\begin{align}
    f(x) = \braket{x}{f} = \braket{f}{x}
    \, .
\end{align}
In this formalism, the delta function is given by
\begin{align}
    \braket{x}{y} = \delta(x - y)
    \, .
\end{align}
The identity operator is constructed by the complete set,
\begin{align}
    1 = \int \dd{x} \ket{x} \bra{x} 
    \, .
\end{align}
The trace of the operator is given by
\begin{align}
    \Tr \mathcal{O} = \int \dd{x} \bra{x} \mathcal{O} \ket{x} = \int \dd{x} \mathcal{O}(x,x)
    \, .
\end{align}

\paragraph{Coordinate and differential operators}

We define the coordinate and differential operators,
\begin{align}
    x f(x) = \bra{x} Q \ket{f} = \bra{f} Q \ket{x}
    \, , \qquad
    f'(x) = \bra{x} P \ket{f} = - \bra{f} P \ket{x}
    \, ,
\end{align}
whose kernels are given by
\begin{align}
    \bra{x} Q \ket{y} = x \delta(x - y)
    \, , \qquad
    \bra{x} P \ket{y} = \pdv{}{x} \delta(x - y)
    \, .
\end{align}

\paragraph{Kernels and wave functions}

We define the operator $K$ corresponding to the $p$-Airy kernel,
\begin{align}
    K(x,y) \equiv K_{p\text{-Airy}}(x,y) = \bra{x} K \ket{y}   
    \, .
\end{align}
We then introduce the states associated with the wave functions,
\begin{align}
    \phi_k(x) = \dv[k]{\phi(x)}{x} =  \braket{x}{\phi_k}
    \, , \qquad
    \psi_k(x) = \dv[k]{\psi(x)}{x} = \braket{\psi_k}{x}
    \, ,
\end{align}
and we put
\begin{align}
    \ket{\phi_0} = \ket{\phi}
    \, ,  \qquad
    \ket{\psi_0} = \ket{\psi}
    \, .
\end{align}
For even $p$, they are equivalent to each other, $\ket{\phi_k} = \ket{\psi_k}$.
We remark the relation
\begin{align}
    P \ket{\phi_k} = \ket{\phi_{k+1}}
    \, , \qquad
    P \ket{\psi_k} = \ket{\psi_{k+1}}
%    \quad \qty( \iff - \bra{\psi_k} P = \bra{\psi_{k+1}} \ ) 
 \, ,
\end{align}
and also
\begin{align}
 \bra{\psi_k} = (-1)^k \bra{\psi} P^k
 \, .
 \label{eq:psi_k_op}
\end{align}

From the relation
\begin{align}
    \bra{x} \comm{Q}{K} \ket{y} = (x-y) K(x,y)
    \, ,
\end{align}
we obtain a useful formula
\begin{align}
    \comm{Q}{K} = \sum_{k=1}^p (-1)^k \ket{\phi_{p-k}}\bra{\psi_{k-1}}
    \, .
\end{align}

\subsection{Fredholm determinant}\label{sec:Fredholm_det}

We consider the intervals $I = \bigsqcup_{j=1}^m [a_{2j-1},a_{2j}] \subset \mathbb{R}$.
We denote the parameter set by $\underline{a} = (a_j)_{j = 1,\ldots,2m}$.
\begin{definition}
Let $\chi_I(\cdot)$ be the characteristic function associated with the interval $I$:
\begin{align}
    \chi_I(x) = 
    \begin{cases}
    1 & (x \in I) \\ 0 & (x \not\in I)
    \end{cases}
\end{align}
and the corresponding projection operator 
\begin{align}
 \bra{x} \Pi_I \ket{f} = f(x) \chi_I(x)
 \, , \qquad \Pi_I^2 = \Pi_I
 \, .
 \label{eq:projection_op}
\end{align}
We define the restricted kernels and the corresponding operators
\begin{subequations}\label{eq:K_op_I}
\begin{align}
    \bra{x} \hat{K} \ket{y} = K(x,y) \chi_I(y)
    \, , \qquad
    \hat{K} = K \Pi_I
    \, , \\
    \bra{x} \check{K} \ket{y} = \chi_I(x) K(x,y)
    \, , \qquad
    \check{K} = \Pi_I K 
    \, .
\end{align}
\end{subequations}
\end{definition}

 \begin{lemma}
  We have the following commutation relation,
 \begin{align}
  \comm{\Pi_I}{P} = \sum_{j=1}^{2m} (-1)^j \ket{a_j} \bra{a_j}
  \, .
  \label{eq:Pi_P_commutator}
 \end{align}
 \end{lemma}
  \begin{proof}
   Taking the derivative of the expression \eqref{eq:projection_op}, we obtain
  \begin{align}
   \bra{x} P \Pi_I \ket{f}
   & = \dv{}{x} \qty( f(x) \chi_I(x) )
   \nonumber \\
   & = f'(x) \chi_I(x) - \sum_{j=1}^{2m} (-1)^j f(a_j) \delta(x - a_j)
   \nonumber \\
   & = \bra{x}\Pi_I P \ket{f} - \sum_{j=1}^{2m} (-1)^j \braket{x}{a_j} \braket{a_j}{f}
   \, ,
  \end{align}
   which yields the relation \eqref{eq:Pi_P_commutator}.
  \end{proof}

We consider the Fredholm determinant associated with the $p$-Airy kernel,
\begin{align}
    %F(I) = 
    F(\underline{a}) %= F_p(\underline{a}) 
    = \det ( 1 - \hat{K} ) = \det\qty( 1 - \check{K} )
    \, .
    \label{eq:Fredholm_det}
\end{align}
Taking the logarithm of the Fredholm determinant, we obtain
\begin{align}
    \log F(\underline{a}) = - \sum_{n=1}^\infty \frac{1}{n} \Tr \hat{K}^n = - \sum_{n=1}^\infty \frac{1}{n} \Tr \check{K}^n
    \, ,
    \label{eq:Fredholm_det_log}
\end{align}
where we have
\begin{align}
    \Tr \hat{K}^n = \Tr \check{K}^n & = \int_{I} \dd{x_1} \cdots \int_I \dd{x_n} K(x_1,x_2) K(x_2,x_3) \cdots K(x_{n},x_1)
    \, .
\end{align}

\subsubsection{Resolvent operators}

We define the resolvent operators from the kernels,
\begin{align}
    R = \sum_{n=1}^\infty \hat{K}^n = \frac{\hat{K}}{1 - \hat{K}}
    \, , \qquad
    L = \sum_{n=1}^\infty \check{K}^n = \frac{\check{K}}{1 - \check{K}}
    \, .
\end{align}
One can show the following formulas for the resolvent,
\begin{subequations}
\begin{align}
    \comm{X}{(1 - \hat{K})^{-1}} & = (1 - \hat{K})^{-1} \comm{X}{\hat{K}} (1 - \hat{K})^{-1}
    \, , \\
    \partial (1 - \hat{K})^{-1} & = (1 - \hat{K})^{-1} \qty(\partial \hat{K}) (1 - \hat{K})^{-1} 
    \, .
\end{align}
\end{subequations}
The same relation holds for $\check{K}$.
Then, the parameter dependence of the Fredholm determinant is given as follows.
\begin{lemma}
The total derivative of the Fredholm determinant is
\begin{align}
    \dd{ \log F(\underline{a}) } 
    & = \sum_{j=1}^{2m} (-1)^{j-1} R(a_j,a_j) \dd{a_j}
    = \sum_{j=1}^{2m} (-1)^{j-1} L(a_j,a_j) \dd{a_j}
    \nonumber \\
    & = \sum_{j=1}^{2m} (-1)^{j-1} H_j \dd{a}_j
    \label{eq:logF_dependence}
\end{align}
where we define the Hamiltonians,
\begin{align}
    H_j = % H_j(a_j) = 
    R(a_j,a_j) = L(a_j,a_j)
    \, , \qquad
    j = 1, \ldots, 2m
    \, .
\end{align}
\end{lemma}
\begin{proof}
The parameter derivative of the projection operator is given by
\begin{align}
    \pdv{}{a_j} \Pi_I = (-1)^j \ket{a_j} \bra{a_j}
    \, ,
\end{align}
which yields
%\begin{subequations}
\begin{align}
    \pdv{}{a_j} \hat{K} = (-1)^j K \ket{a_j} \bra{a_j}
    \, , \qquad
    \pdv{}{a_j} \check{K} = (-1)^j \ket{a_j} \bra{a_j} K 
    \, .     
    \label{eq:K_derivative_aj}
\end{align}
%\end{subequations}
Then, the parameter derivative of the Fredholm determinant is given by
\begin{subequations}
\begin{align}
    \pdv{}{a_j} \log \det (1 - \hat{K}) & = - \Tr \qty(1 - \hat{K})^{-1} \pdv{\hat{K}}{a_j}
    \nonumber \\
    & = (-1)^{j-1} \Tr R \ket{a_j} \bra{a_j} = (-1)^{j-1} R(a_j,a_j)
    \, , \\
    \pdv{}{a_j} \log \det (1 - \check{K}) & = - \Tr \qty(1 - \check{K})^{-1} \pdv{\check{K}}{a_j}
    \nonumber \\
    & = (-1)^{j-1} \Tr \ket{a_j} \bra{a_j} L = (-1)^{j-1} L(a_j,a_j)    
    \, .
\end{align}
\end{subequations}
Taking into account all the contributions of the parameters $(a_j)_{j=1,\ldots,2m}$, we obtain~\eqref{eq:logF_dependence}.
\end{proof}

\begin{remark}
For $x,y \in I$, we have
\begin{align}
    R(x,y) = L(x,y)
    \, .
\end{align}
Hence, we have
\begin{align}
    R(a_i,a_j) = L(a_i,a_j)
    \, ,  \qquad
    i,j = 1,\ldots,2m \, ,
    \label{eq:Rij=Lij}
\end{align}
so that we denote
\begin{align}
    R_{ij} := R(a_i,a_j) = L(a_i,a_j) \, .
\end{align}
\end{remark}

The expression~\eqref{eq:logF_dependence} means that the parameter dependence of the Fredholm determinant is fixed by the diagonal values of the resolvent kernels, that we call the Hamiltonians.
In fact, the Fredholm determinant plays a role of the isomonodromic $\tau$-function, and the parameters $(a_j)_{j=1,\ldots,2m}$ are interpreted as the corresponding time variables in this context~\cite{Jimbo:1981zz,Jimbo:1981II,Jimbo:1981III}.
See also~\cite{Harnard:2021}.

\subsubsection{Auxiliary wave functions}

As shown above, we shall evaluate the Hamiltonians to describe the Fredholm determinant.
For this purpose, we introduce auxiliary wave functions, $(\underline{\sfq},\underline{\sfp}) = (\sfq_k,\sfp_k)_{k=0,\ldots,p-1}$,
\begin{subequations}
\begin{align}
    \sfq_k(x) = \bra{x} (1 - \hat{K})^{-1} \ket{\phi_k}
    \, , \qquad
    \tilde{\sfq}_k(x) = \bra{x} (1 - \check{K})^{-1} \Pi_I \ket{\phi_k} = \sfq_k(x) \chi_I(x)
    \, , \\
    \sfp_k(x) = \bra{\psi_k} (1 - \check{K})^{-1} \ket{x}
    \, , \qquad
    \tilde{\sfp}_k(x) = \bra{\psi_k} \Pi_I (1 - \hat{K})^{-1} \ket{x} = \sfp_k(x) \chi_I(x)
    \, .
\end{align}
\end{subequations}
Using the resolvent kernels, we may write them as follows,
\begin{subequations}\label{eq:qp_RL}
\begin{align}
    \sfq_k(x) & = \phi_k(x) + \bra{x} R \ket{\phi_k}
    \, , \\
    \sfp_k(x) & = \psi_k(x) + \bra{\psi_k} L \ket{x}
    \, .
\end{align}
\end{subequations}
Then, we obtain the expression of the resolvent kernels using these auxiliary wave functions.
\begin{lemma}[Resolvent kernel via the auxiliary wave functions]
The resolvent kernels are written in terms of the auxiliary wave functions,
\begin{subequations}\label{eq:RL_xy}
\begin{align}
    R(x,y) & = \frac{1}{x - y} \sum_{k=1}^p (-1)^k \sfq_{p-k}(x) \tilde{\sfp}_{k-1}(y)
    \, , \\
    L(x,y) & = \frac{1}{x - y} \sum_{k=1}^p (-1)^k \tilde{\sfq}_{p-k}(x) {\sfp}_{k-1}(y)
    \, .
\end{align}
\end{subequations}
\end{lemma}
\begin{proof}
This expression is obtained from the relation,
\begin{align}
    (x-y) \bra{x} R \ket{y} 
    & = \bra{x} \comm{Q}{R} \ket{y}
    \nonumber \\
    & = \bra{x} (1 - \hat{K})^{-1} \comm{Q}{\hat{K}} (1 - \hat{K})^{-1} \ket{y}
    \nonumber \\
    & = \sum_{k=1}^p (-1)^k \bra{x} (1 - \hat{K})^{-1} \ket{\phi_{p-k}} \bra{\psi_{k-1}} \Pi_I (1 - \hat{K})^{-1}  \ket{y}
    \nonumber \\
    & = \sum_{k=1}^p (-1)^k \sfq_{p-k}(x) \tilde{\sfp}_{k-1}(y)
    \, .
\end{align}
We may apply the same analysis to the other kernel $L(x,y)$.
\end{proof}

\begin{remark}
For $x \in I \ \qty(\iff \chi_I(x) = 1)$, the diagonal value of the resolvent kernel is given by
\begin{align}
    R(x,x) = L(x,x) & = \sum_{k=1}^p (-1)^k \sfq_{p-k}'(x) \sfp_{k-1}(x) = \sum_{k=1}^p (-1)^{k-1} \sfq_{p-k}(x) \sfp_{k-1}'(x)
    \, .
    \label{eq:RL_diagonal}
\end{align}
See also Lemma~\ref{lem:pq_rel}.
\end{remark}

\begin{lemma}[Parameter dependence of $(\sfq_k,\sfp_k)$]\label{lem:pq_par_dependence}
The auxiliary wave functions show the following parameter dependence,
\begin{align}
    \pdv{\sfq_k}{a_j} = (-1)^{j} R(x,a_j) \sfq_k(a_j)
    \, ,  \qquad
    \pdv{\sfp_k}{a_j} = (-1)^{j} \sfp_k(a_j) L(a_j,x) 
    \, .
\end{align}
\end{lemma}
\begin{proof}
We notice the relation
\begin{align}
    \ket{a_j} = \Pi_I \ket{a_j}
    \, ,
    \label{eq:Pi|a>}
\end{align}
so that $K \ket{a_j} = \hat{K} \ket{a_j}$ and $\bra{a_j} K = \bra{a_j} \check{K}$.
Then, we obtain as follows,
\begin{align}
    \pdv{\sfq_k}{a_j}
    & = \bra{x} (1 - \hat{K})^{-1} {\pdv{ \hat{K} }{a_j} }
     (1 - \hat{K})^{-1} \ket{\phi_k}
    \nonumber \\
    & = (-1)^j \bra{x} (1 - \hat{K})^{-1} K \ket{a_j} \bra{a_j} (1 - \hat{K})^{-1} \ket{\phi_k}
    \nonumber \\
    & = (-1)^{j} R(x,a_j) \sfq_k(a_j)
    \, .
\end{align}
The derivatives of $\sfp_k$ are similarly obtained.
\end{proof}

Then, from the regularity of the resolvent kernel at $x = y$, we have the following relation, which is an analog of the relation \eqref{eq:diagonal_id}.
\begin{lemma}\label{lem:pq_rel}
The following relation holds for the auxiliary wave functions,
\begin{align}
    \sum_{k=1}^p (-1)^k \sfq_{p-k}(x) \sfp_{k-1}(x) = 0
    \, .
\end{align}
\end{lemma}
\begin{proof}
Applying the parameter derivative of the auxiliary functions shown in Lemma~\ref{lem:pq_par_dependence}, we can show that there is no parameter dependence,
 \begin{align}
  \pdv{}{a_j} \sum_{k=1}^p (-1)^k \sfq_{p-k}(x) \sfp_{k-1}(x) = 0
  \, , \qquad
  j = 1,\ldots,2m
  \, .
 \end{align}
 Hence, we can modify them to obtain $I = \emptyset$, where all the auxiliary functions become zero.
 Therefore, the constant of integration turns out to be zero.
\end{proof}

%\paragraph{$x$-dependence of $(\sfq_k,\sfp_k)$}
\begin{lemma}[$x$-dependence of $(\sfq_k,\sfp_k)$]
The derivatives of the auxiliary wave functions are given as follows,
\begin{subequations}\label{eq:der_pa}
\begin{align}
    \pdv{}{x} \sfq_k(x)
    & = \sfq_{k+1}(x) - \sfq_0(x) u_k - \sum_{j=1}^{2m} (-1)^j R(x,a_j) \sfq_k(a_j) 
    \, , \\
    \pdv{}{x} \sfp_k(x)
    & = \sfp_{k+1}(x) - v_k \sfp_0(x) - \sum_{j=1}^{2m} (-1)^j \sfp_k(a_j) L(a_j,x) 
    \, ,
\end{align}
\end{subequations}
where we define auxiliary functions $(\underline{u},\underline{v}) = (u_k,v_k)_{k=0,\ldots,p-1}$ of the parameters $(a_j)_{j = 1,\ldots,2m}$,
\begin{subequations}\label{eq:uv_integral}
\begin{align}
 u_k & = \bra{\psi} \Pi_I (1 - \hat{K})^{-1} \ket{\phi_k}
 = \int_I \dd{x} \psi(x) \phi_k(x) + \int_I \dd{x} \int_I \dd{y} \psi(x) R(x,y) \phi_k(y)
 \, , \\
 v_k & = \bra{\psi_k} (1 - \check{K})^{-1} \Pi_I \ket{\phi}
 = \int_I \dd{x} \psi_k(x) \phi(x) + \int_I \dd{x} \int_I \dd{y} \psi_k(x) R(x,y) \phi_(y)
 \, .
\end{align}
\end{subequations}
We remark
\begin{align}
    u_0 = v_0 \, .
    \label{eq:u0=v0}
\end{align}
\end{lemma}
\begin{proof}
In order to show these expressions, we use the formula
\begin{align}
    \bra{x} \comm{P}{\hat{K}} \ket{y} 
%    & = \bra{x} P K \Pi_I - K \Pi_I P \ket{y} 
%    \nonumber \\
%    & = \bra{x} P K \Pi_I - (K P) \Pi_I - K (\Pi_I P) \ket{y}
%    \nonumber \\
    & = \qty( \pdv{}{x} + \pdv{}{y} ) K(x,y) \chi_I(y) % + \sum_{j=1}^m \qty( - K(x,a_j) \delta(a_j - y) + K(x,b_j) \delta(b_j - y) )
    \nonumber \\
    & = - \phi(x) \psi(y) \chi_I(y) - \sum_{j=1}^{2m} (-1)^{j} K(x,y) \delta(a_j - y)
    \nonumber \\
    & = \bra{x} \qty[- \ket{\phi}\bra{\psi} \Pi_I - \sum_{j=1}^{2m} (-1)^{j} K \ket{a_j} \bra{a_j}  ]\ket{y}
\end{align}
where we have
\begin{align}
    \qty( \pdv{}{x} + \pdv{}{y} ) K(x,y) = - \phi(x) \psi(x)
    \quad \iff \quad
    \comm{P}{K} = - \ket{\phi}\bra{\psi}
    \, .
    \label{eq:P_K_commutator}
\end{align}
We may also obtain this expression by writing 
\begin{align}
 \bra{x} \comm{P}{\hat{K}} \ket{y}
 & = \bra{x} \comm{P}{K} \Pi_I + K \comm{P}{\Pi_I} \ket{y}
 \, ,
\end{align}
with the relation \eqref{eq:Pi_P_commutator}.
A similar relation holds for $\check{K}$,
\begin{align}
    \bra{x} \comm{P}{\check{K}} \ket{y} & = \qty( \pdv{}{x} + \pdv{}{y} ) \chi_I(x) K(x,y)  
    \nonumber \\
    & =
    \bra{x} \qty[- \Pi_I \ket{\phi}\bra{\psi}  - \sum_{j=1}^{2m} (-1)^j \ket{a_j} \bra{a_j} K ]\ket{y}
    \, .
\end{align}
Then, we compute
\begin{subequations}
\begin{align}
    \pdv{}{x} \sfq_k(x)
    & = \bra{x} P (1 - \hat{K})^{-1} \ket{\phi_k}
    \nonumber \\
    & = \bra{x} (1 - \hat{K})^{-1} \comm{P}{\hat{K}} (1 - \hat{K})^{-1} + (1 - \hat{K})^{-1} P \ket{\phi_k}
    \nonumber \\
    & = \sfq_{k+1}(x) - \sfq_0(x) u_k - \sum_{j=1}^{2m} (-1)^j R(x,a_j) \sfq_k(a_j) 
    \, , \\
    \pdv{}{x} \sfp_k(x)
    & = - \bra{\psi_k} (1 - \check{K})^{-1} P \ket{x}
    \nonumber \\
    & = - \bra{\psi_k} (1 - \check{K})^{-1} \comm{\check{K}}{P} (1 - \check{K})^{-1} + P (1 - \check{K})^{-1} \ket{x}
    \nonumber \\
    & = \sfp_{k+1}(x) - v_k \sfp_0(x) - \sum_{j=1}^{2m} (-1)^j \sfp_k(a_j) L(a_j,x) 
    \, ,
\end{align}
\end{subequations}
to obtain the formulas \eqref{eq:der_pa}.
\end{proof}

\paragraph{Parameter dependence of $(u_k,v_k)$}
We consider the parameter dependence of $(u_k,v_k)_{k = 1,\ldots,p}$.
Since we have
\begin{align}
    \pdv{}{a_j} \qty( \Pi_I (1 - \hat{K})^{-1} )
    & = (-1)^j \ket{a_j} \bra{a_j} (1 - \hat{K})^{-1} + \Pi_I (1 - \hat{K})^{-1} {\pdv{\hat{K}}{a_j}} (1 - \hat{K})^{-1}
    \nonumber \\
    & = (-1)^j \Pi_I (1 - \hat{K})^{-1} \ket{a_j} \bra{a_j} (1 - \hat{K})^{-1}
\end{align}
with \eqref{eq:Pi|a>}, we obtain
%\begin{subequations}
\begin{align}
    \pdv{u_k}{a_j} = (-1)^j \sfq_k(a_j) \sfp_0(a_j)
    \, , \qquad
    \pdv{v_k}{a_j} = (-1)^j \sfq_0(a_j) \sfp_k(a_j) 
    \, .
    \label{eq:uv_dependence}
\end{align}
%\end{subequations}

\paragraph{Rewriting $(\sfq_p,\sfp_p)$}

From the differential equations for the wave functions~\eqref{eq:Airy_ODE}, we have the relations,
\begin{align}
    \ket{\phi_p} = Q \ket{\phi_0}
    \, , \qquad
    \bra{\psi_p} = (-1)^p \bra{\psi_0} Q
    \, .
\end{align}
Then, we have
\begin{align}
    \sfq_p(x) & = \bra{x} (1 - \hat{K})^{-1} \ket{\phi_p}
    \nonumber \\
    & = \bra{x} (1 - \hat{K})^{-1} Q \ket{\phi_0}
    \nonumber \\
    & = \bra{x} Q (1 - \hat{K})^{-1} \ket{\phi_0} - \bra{x} (1 - \hat{K})^{-1} \comm{Q}{\hat{K}} (1 - \hat{K})^{-1} \ket{\phi_0}
    \nonumber \\
    & = x \sfq_0(x) - \sum_{k=1}^p (-1)^k \bra{x} (1 - \hat{K})^{-1} \ket{\phi_{p-k}} \bra{\psi_{k-1}} \Pi_I (1 - \hat{K})^{-1} \ket{\phi_0}
%    \nonumber \\
%    & = x \sfq_0(x) - \sum_{k=1}^n (-1)^k \sfq_{p-k}(x) \, v_{k-1}
    \, ,    
\end{align}
which leads to the expressions
\begin{subequations}\label{eq:n_lower}
\begin{align}
    \sfq_p(x) & = x \sfq_0(x) - \sum_{k=1}^p (-1)^k \sfq_{p-k}(x) \, v_{k-1}
    \, , \\
    \sfp_p(x) & = (-1)^p x \sfp_0(x) - \sum_{k=1}^p (-1)^k u_{k-1} \, \sfp_{p-k}(x)
    \, .
\end{align}
\end{subequations}

\paragraph{Rewriting $(u_p,v_p)$}

Similarly, we can rewrite $(u_p,v_p)$ as follows,
\begin{subequations}\label{eq:n_lower_uv}
\begin{align}
    u_p & = \int_I \dd{x} x \, \psi(x) \sfq_0(x) - \sum_{k=1}^p (-1)^k u_{p-k} \, v_{k-1}
    \, , \\
    v_p & = (-1)^p \int_I \dd{x} x \, \sfp_0(x) \phi(x) -  \sum_{k=1}^p (-1)^k u_{k-1} \, v_{p-k}
    \, .
\end{align}
\end{subequations}
We remark the relation,
\begin{align}
    u_p - (-1)^p v_p & = - \sum_{k=1}^p (-1)^k u_{p-k} \, v_{k-1}
    \, .
    \label{eq:up+vp}
\end{align}

\subsubsection{Hamiltonian system}

We denote 
\begin{align}
    \sfq_{k,j} = \sfq_k(a_j)
    \, , \qquad
    \sfp_{k,j} = \sfp_k(a_j)
    \, .
\end{align}
The parameter dependence of the wave functions $(\sfq_{k,j},\sfp_{k,j})$ is given by
\begin{subequations}\label{eq:qp_dependence}
\begin{align}
    \dv{\sfq_{k,j}}{a_j} & = \qty( \pdv{}{x} + \pdv{}{a_j} ) \sfq_k(x)\Bigg|_{x \to a_j}
%    \nonumber \\ &
    = \sfq_{k+1,j} - \sfq_{0,j} u_k - \sum_{\ell (\neq j)}^{2m} (-1)^\ell R_{j\ell} \sfq_{k,\ell}
    \, , \\
    \dv{\sfp_{k,j}}{a_j} & = \qty( \pdv{}{x} + \pdv{}{a_j} ) \sfp_k(x)\Bigg|_{x \to a_j}
%    \nonumber \\ &
    = \sfp_{k+1,j} - \sfp_{0,j} v_k - \sum_{\ell (\neq j)}^{2m} (-1)^\ell \sfp_{k,\ell} R_{\ell j}
    \, .
\end{align}
\end{subequations}
Therefore, we have the following expression for the Hamiltonians.
\begin{proposition}[Hamiltonian system]\label{prop:Ham_sys}
The Hamilsonians associated with the Fredholm determinant are written as follows,
\begin{align}
    H_j & = R_{jj} %R(a_j,a_j) = L(a_j,a_j)  
    \nonumber \\
    & = \sum_{k=1}^p (-1)^k \pdv{\sfq_{p-k,j}}{a_j} \sfp_{k-1,j}
    = \sum_{k=1}^p (-1)^{k-1} \sfq_{p-k,j} \pdv{\sfp_{k-1,j}}{a_j}
    \, ,
    \label{eq:Hamiltonian_Rjj}
\end{align}
whose explicit form is given by
\begin{align}
    H_j & = - a_j \sfq_{0,j} \sfp_{0,j} - \sum_{k=1}^{p-1} (-1)^k \sfq_{p-k,j} \sfp_{k,j} 
    \nonumber \\ & \qquad 
    + \sum_{k=1}^p (-1)^k \qty( \sfq_{p-k,j} \sfp_{0,j} v_{k-1} - \sfq_{0,j} \sfp_{k-1,j} u_{p-k} ) 
    - \sum_{\ell(\neq j)}^{2m} (-1)^\ell (a_\ell - a_j) R_{j\ell} R_{\ell j}
    \, .
    \label{eq:Hj}
\end{align}
\end{proposition}

\begin{corollary}[Hamilton equations]\label{cor:Ham_eqs}
From the Hamiltonians obtained above, we have the Hamilton equations for the wave functions,
\begin{subequations}
\begin{align}
    \pdv{\sfq_{p-k,j}}{a_j} & = (-1)^k \pdv{H_j}{\sfp_{k-1,j}}
    \, ,  &
    \pdv{\sfp_{k-1,j}}{a_j} & = (-1)^{k-1} \pdv{H_j}{\sfq_{p-k,j}}
    \, ,
    \label{eq:Hamilton_eq}
    \\
    \pdv{u_{p-k,j}}{a_j} & = (-1)^{j+k} \pdv{H_j}{v_{k-1,j}}
    \, ,  &
    \pdv{v_{k-1,j}}{a_j} & = (-1)^{j+k-1} \pdv{H_j}{u_{p-k,j}}
    \, .
\end{align}
\end{subequations}
This shows that $(\sfq_{p-k},\sfp_{k-1})_{k = 1,\ldots,p}$ and $(u_{p-k},v_{k-1})_{k = 1,\ldots,p}$ form the canonical pairs.
For even $p$, $\sfq_k$ and $\sfp_k$ ($u_k$ and $v_k$ as well) are equivalent.
Hence, $(\sfq_{p-k},\sfq_{k-1})_{k = 1,\ldots,p}$ (and $(u_{p-k},u_{k-1})_{k = 1,\ldots,p}$) are the canonical pairs with respect to the time variables $\underline{a} = (a_j)_{j=1,\ldots,2m}$.
\end{corollary}

\begin{lemma}\label{lem:Ham_derivative}
The parameter dependence of the Hamiltonians is given by
\begin{align}
%    \pdv{}{c_j} R(c_j,c_j) = \pdv{}{c_j} L(c_j,c_j) 
    \dv{H_j}{a_j}
    & = - \sfq_{0,j} \sfp_{0,j} - \sum_{\ell (\neq j)}^{2m} (-1)^\ell R_{j\ell} R_{\ell j}
%    \, , \nonumber \\
%    & = - \sfq_{0,j} \sfp_{0,j} - \sum_{\ell (\neq j)}^{2m} (-1)^\ell R_{j\ell}^2
    \, .
    \label{eq:H_dependence}
\end{align}
\end{lemma}
\begin{proof}
This is obtained as follows,
\begin{align}
    \dv{H_{j}}{a_j} = \dv{R_{jj}}{a_j} 
    & = \dv{}{a_j} \qty( \bra{a_j} \frac{\hat{K}}{1 - \hat{K}} \ket{a_j} )
    \nonumber \\
    & = \bra{a_j} (1 - \hat{K})^{-1} \comm{P}{\hat{K}} (1 - \hat{K})^{-1} + (1 - \hat{K})^{-1} \qty( \pdv{\hat{K}}{a_j} ) (1 - \hat{K})^{-1} \ket{a_j}
    \nonumber \\
    & = \bra{a_j} (1 - \hat{K})^{-1} \qty( - \ket{\phi}\bra{\psi} \Pi_I - \sum_{\ell = 1 (\neq j)}^{2m} (-1)^\ell K \ket{a_\ell}\bra{a_\ell} ) (1 - \hat{K})^{-1} \ket{a_j}
    \nonumber \\
    & = - \sfq_{0,j} \sfp_{0,j} - \sum_{\ell = 1 (\neq j)}^{2m} (-1)^\ell R_{j\ell} R_{\ell j}
\end{align}
where we use the relation
\begin{align}
    \bra{a_\ell} (1 - \hat{K})^{-1} \ket{a_j}
    = \bra{a_\ell} \frac{\hat{K}}{1 - \hat{K}} \ket{a_j} + \braket{a_\ell}{a_j}
    \, .
\end{align}
\end{proof}

\begin{lemma}[Integral of motion]\label{lem:IOM}
The integrals of motion defined as
\begin{align}
    \mathsf{I}_\ell = u_\ell - (-1)^\ell v_\ell + \sum_{k=1}^\ell (-1)^k \qty( u_{\ell - k} v_{k-1} + \sum_{j=1}^{2m} (-1)^j \sfq_{\ell-k,j} \sfp_{k-1,j} )
    \, , \quad
    \ell = 1,\ldots,p \, ,
\end{align}
are zero.
\end{lemma}
\begin{proof}
From Lemma~\ref{lem:pq_par_dependence}, for $i \ne j$, we have
\begin{align}
    \pdv{\sfq_{k,i}}{a_j} = (-1)^{j} R_{ij} \sfq_{k,j}
    \, , \qquad
    \pdv{\sfp_{k,i}}{a_j} = (-1)^{j} \sfp_{k,j} R_{ji}
    \, .
    \label{lem:pq_par_dependence_ij}
\end{align}
Using the relations \eqref{eq:qp_dependence} and \eqref{eq:uv_dependence}, the derivative of the integral of motion with the parameter $a_j$ is given by
\begin{align}
    (-1)^j \dv{\mathsf{I}_\ell}{a_j} & = 
    \sfq_{\ell,j} \sfp_{0,j} - (-1)^\ell \sfq_{0,j} \sfp_{\ell,j} + \sum_{k = 1}^\ell (-1)^k \qty( \sfq_{\ell-k+1,j} \sfp_{k-1,j} + \sfq_{\ell-k,j} \sfp_{k,j} )
    \nonumber \\
    & = 0
    \, .
\end{align}
Since it does not depend on the parameters $\underline{a} = (a_j)_{j = 1,\ldots,2m}$, we can modify them to obtain $I = \emptyset$, which yields $\sfq_k, \sfp_k, u_k, v_k = 0$.
\end{proof}

We can also derive the integrals of motion in the operator formalism.
Recalling
\begin{align}
 \comm{\hat{K}}{P} = K \comm{\Pi_I}{P} + \comm{K}{P} \Pi_I
 \, ,
\end{align}
and
\begin{align}
 (1-\hat{K})^{-1} \hat{K} + 1 = (1-\hat{K})^{-1}
 \, ,
\end{align}
together with the relations \eqref{eq:Pi_P_commutator} and \eqref{eq:P_K_commutator}, we have
\begin{align}
 u_\ell
 & = \bra{\psi} \Pi_I (1 - \hat{K})^{-1} P^\ell \ket{\phi}
 \nonumber \\
 & = \bra{\psi} \Pi_I \qty[ (1-\hat{K})^{-1} \comm{\hat{K}}{P} (1-\hat{K})^{-1} + P (1-\hat{K})^{-1} ] P^{\ell-1} \ket{\phi}
 \nonumber \\
 & = \bra{\psi} \Pi_I (1-\hat{K})^{-1} \qty( \comm{\Pi_I}{P} + \comm{K}{P} \Pi_I ) (1-\hat{K})^{-1} P^{\ell-1} \ket{\phi}
 + \bra{\psi} P \Pi_I (1 - \hat{K})^{-1} P^{\ell-1} \ket{\phi}
 \nonumber \\
 & =
 v_0 u_{\ell-1} + \sum_{j=1}^{2m} (-1)^j \sfq_{\ell-1,j} \sfp_{0,j}
 + \bra{\psi} P \Pi_I (1 - \hat{K})^{-1} P^{\ell-1} \ket{\phi}
 \, .
\end{align}
Applying this process recursively together with the relation \eqref{eq:psi_k_op}, we arrive at Lemma~\ref{lem:IOM}.

In particular, applying the expression \eqref{eq:up+vp}, the integral of motion for $\ell = p$ is given by
\begin{align}
    \mathsf{I}_p & = \sum_{k=1}^p \sum_{j=1}^{2m} (-1)^{k+j} \sfq_{p-k,j} \sfp_{k-1,j}
    \, .
    \label{eq:IOM_highest}
\end{align}
We can show also from Lemma~\ref{lem:pq_rel} that this integral of motion becomes zero.

\subsubsection{Lax formalism}

We introduce the vector notation
\begin{align}
    \mathsf{Q} = (\sfq_0 \ \sfq_1 \ \cdots \ \sfq_{p-1})^\text{T}
    \, , \qquad
    \mathsf{P} = (\sfp_0 \ \sfp_1 \ \cdots \ \sfp_{p-1})
    \, ,
\end{align}
where we denote
\begin{align}
    \sfq_k(x) = \sfq_k \, , \qquad \sfp_k(x) = \sfp_k \, .
\end{align}
We also denote $(\sfq,\sfp) = (\sfq_0,\sfp_0)$.
Then, we may express the parameter dependence in the matrix form as follows.
\begin{lemma}[Lax formalism]
The $x$-dependence and the parameter dependence of the auxiliary wave functions are given as follows,
\begin{subequations}
\begin{align}
    \pdv{\mathsf{Q}}{x} & = \qty( \sum_{j=1}^{2m} \frac{A_j}{x-a_j} + A_\infty ) \mathsf{Q}
    \, , \qquad
    \pdv{\mathsf{Q}}{a_j} = - \frac{A_j}{x-a_j} \mathsf{Q}
    \, , \\
    \pdv{\mathsf{P}}{x} & = \mathsf{P} \qty( \sum_{j=1}^{2m} \frac{\widetilde{A}}{x-a_j} + \widetilde{A}_\infty )
    \, , \qquad
    \pdv{\mathsf{P}}{a_j} = - \mathsf{P} \frac{\widetilde{A}_j}{x-a_j} 
    \, .
\end{align}
\end{subequations}
We call this the Lax matrix, whose coefficients $(A_j,\widetilde{A}_j)_{j = 1,\ldots,2m,\infty}$ are given by
\begin{subequations}
\begin{align}
    (A_j)_{k,k'} = (-1)^{j+p-k'+1} \sfq_{k,j} \sfp_{p-k'-1,j}
    \, ,  \qquad
    (\widetilde{A}_j)_{k,k'} = (-1)^{j+p-k+1} \sfq_{p-k-1,j} \sfp_{k',j} 
    \, ,
\end{align}
\begin{align}
    A_\infty & = 
    \begin{pmatrix}
    - u_0 & 1 & & & \\
    - u_1 & 0 & 1 && \\
    \vdots&&\ddots&\ddots& \\
    - u_{p-2} &&& 0 & 1 \\
    x - u_{p-1} + (-1)^{p-1} v_{p-1} & (-1)^{p-2} v_{p-2} & \cdots & - v_1 & v_0
    \end{pmatrix}
    \, , \\[.5em]
    \widetilde{A}_\infty & = 
    \begin{pmatrix}
    - v_0 & - v_1 & \cdots & - v_{p-2} & (-1)^p x + (-1)^{p-1} u_{p-1} - v_{p-1} \\
    1 & 0 &&& (-1)^{p-2} u_{p-2} \\
    &1&\ddots&&\vdots \\
    && \ddots & 0 & - u_1 \\
    &&& 1 & u_0
    \end{pmatrix}
    \, .
\end{align}
\end{subequations}
\end{lemma}
\begin{proof}
From the expression \eqref{eq:der_pa} together with \eqref{eq:RL_xy}, we obtain the $x$-derivative for $k = 0,\ldots,p-2$,
\begin{subequations}
\begin{align}
    \pdv{\sfq_k}{x} & = \sfq_{k+1} - u_k \sfq - \sum_{j=1}^{2m} \sum_{\ell = 1}^p (-1)^{j+\ell} \frac{\sfq_{k,j} \sfp_{\ell - 1,j}}{x - a_j} \sfq_{p-\ell}
    \, , \\
    \pdv{\sfp_k}{x} & = \sfp_{k+1} - v_k \sfp - \sum_{j=1}^{2m} \sum_{\ell = 1}^p (-1)^{j+\ell} \frac{\sfq_{p-\ell,j} \sfp_{k,j}}{x - a_j} \sfp_{\ell-1}
    \, .
\end{align}
\end{subequations}
For $k = p-1$, we may also use \eqref{eq:n_lower} to obtain
\begin{subequations}
\begin{align}
    \pdv{\sfq_{p-1}}{x} & = x \sfq - u_{p-1} \sfq - \sum_{k=1}^p (-1)^k v_{k-1} \sfq_{p-k} - \sum_{j=1}^{2m} \sum_{\ell = 1}^p (-1)^{j+\ell} \frac{\sfq_{p-1,j} \sfp_{\ell - 1,j}}{x - a_j} \sfq_{p-\ell}
    \, , \\
    \pdv{\sfp_{p-1}}{x} & = (-1)^p x \sfp - v_{p-1} \sfp - \sum_{k=1}^p (-1)^k u_{k-1} \sfp_{p-k} - \sum_{j=1}^{2m} \sum_{\ell = 1}^p (-1)^{j+\ell} \frac{\sfq_{p-\ell,j} \sfp_{k,j}}{x - a_j} \sfp_{\ell-1}    
    \, .
\end{align}
\end{subequations}
The parameter dependence immediately follows from Lemma~\ref{lem:pq_par_dependence}.
\end{proof}
\begin{remark}
$A_j^\text{T} = \widetilde{A}_j$ for $j = 1,\ldots,2m$.
\end{remark}
\begin{lemma}
The matrix coefficients $(A_j,\widetilde{A}_j)_{j = 1,\ldots,2m,\infty}$ are traceless.
\end{lemma}
\begin{proof}
It follows from Lemma~\ref{lem:pq_rel} for $j = 1,\ldots,2m$, and from the relation \eqref{eq:u0=v0} for $j = \infty$.
\end{proof}

\begin{lemma}
The total derivative of the Fredholm determinant is given in terms of the matrices $(A_j)_{j = 1,\ldots,2m,\infty}$ as follows,
\begin{align}
    \dd \log F(\underline{a}) = \frac{1}{2} \sum_{i \neq j}^{2m} \tr A_i A_j \frac{\dd a_i - \dd a_j}{a_i - a_j} + \sum_{j = 1}^{2m} \tr A_j A_\infty\Big|_{x = a_j} \dd a_j
    \, .
    \label{eq:dF}
\end{align}
\end{lemma}
\begin{proof}
From the matrices $(A_j)_{j = 1,\ldots,2m,\infty}$, we obtain
\begin{subequations}
\begin{align}
    (-1)^{j-1} \tr A_j A_\ell & = (-1)^{\ell} (a_j - a_\ell)^2 R_{j\ell} R_{\ell j}
    \, , \\
    (-1)^{j-1} \tr A_j A_\infty\Big|_{x = a_j} & = - a_j \sfq_{0,j} \sfp_{0,j} - \sum_{k=1}^{p-1} (-1)^k \sfq_{p-k,j} \sfp_{k,j} 
    \nonumber \\ & \qquad 
    + \sum_{k=1}^p (-1)^k \qty( \sfq_{p-k,j} \sfp_{0,j} v_{k-1} - \sfq_{0,j} \sfp_{k-1,j} u_{p-k} ) 
    \, .
\end{align}
\end{subequations}
Hence, the Hamiltonian \eqref{eq:Hj} is given by
\begin{align}
    (-1)^{j-1} H_j & = \sum_{\ell (\neq j)}^{2m} \frac{\tr A_j A_\ell}{a_j - a_\ell} + \tr A_j A_\infty\Big|_{x = a_j}
    \, .
    \label{eq:Hamltonian_Amatrix}
\end{align}
Substituting this expression to \eqref{eq:logF_dependence}, we obtain \eqref{eq:dF}.
\end{proof}

\begin{proposition}[Schlesinger equation]
The matrix coefficients obey the Schlesinger equation,
\begin{align}
    \pdv{A_i}{a_j} = \frac{\comm{A_i}{A_j}}{a_i - a_j}
    \, , \qquad
    \pdv{A_j}{a_j} = - \sum_{\ell (\neq j)}^{2m} \frac{\comm{A_j}{A_\ell}}{a_j - a_\ell} -\comm{A_j}{A_\infty}\Big|_{x = a_j}
    \, .
    \label{eq:Schlesinger_eq}
\end{align}
\end{proposition}
\begin{proof}
For $i \neq j$, we may use the relation \eqref{lem:pq_par_dependence_ij} to obtain
\begin{align}
    \pdv{(A_i)_{k,k'}}{a_j} & = (-1)^{i+j+p-k'+1} \qty( R_{ij} \sfq_{k,j} \sfp_{p-k'-1,i} + \sfq_{k,i} \sfp_{p-k'-1,j} R_{ji})
    \nonumber \\
    & = \frac{(-1)^{i+j+k'}}{a_i - a_j} \sum_{\ell = 0}^{p-1} (-1)^\ell \qty( \sfq_{k,i} \sfp_{p-\ell-1,i} \sfq_{\ell,j} \sfp_{p-k'-1,j} - (i \leftrightarrow j) )
    \nonumber \\
    & = \frac{(A_i A_j - A_j A_i)_{k,k'}}{a_i - a_j}
    \, .
\end{align}
For $i = j$, we use the Hamilton equations \eqref{eq:Hamilton_eq} to obtain
\begin{align}
    \pdv{(A_j)_{k,k'}}{a_j} & = (-1)^{j-1} \qty( (-1)^{k + k'} \pdv{H_j}{\sfp_{p-k-1,j}} \sfp_{p-k'-1,j} - \sfq_{k,j} \pdv{H_j}{\sfq_{k',j}} )
    \, .
\end{align}
Then, evaluating the derivative of the Hamiltonian with $(\sfq_{k,j}, \sfp_{k,j})$ based on the expression \eqref{eq:Hamltonian_Amatrix}, we arrive at the Schlesinger equation \eqref{eq:Schlesinger_eq}.
\end{proof}

\subsection{Interlude: Comment on unitary matrix integral}

It has been known that the Fredholm determinant associated with the Schur measure kernel studied in \S\ref{sec:Schur_measure} provides an alternative form of the unitary matrix model~\cite{Borodin:2000IEOT}:  
\begin{align}
    \frac{\mathcal{Z}_N}{Z(\mathsf{X},\mathsf{Y})} = \det (1 - K)_{ [N+\frac{1}{2}, \infty) }
\end{align}
where the partition function of the unitary matrix model is given by
\begin{align}
  \mathcal{Z}_N 
  & = \int_{\mathrm{U}(N)} \hspace{-1em} \dd{U} \exp\qty( \sum_{n=1}^\infty \qty( t'_n \tr U^n + \tilde{t}'_n \tr U^{-n}) )
  % \nonumber \\ &
  = \sum_{\lambda_1 \le N} s_\lambda(\mathsf{X}) s_\lambda(\mathsf{Y}).
  \label{pf}
\end{align}
The second equality is obtained by the character expansion with another set of the Miwa variables compared to the previous case~\eqref{eq:Miwa_var},
\begin{align}
    t'_n = \frac{(-1)^{n-1}}{n} \sum_{i=1}^\infty \mathsf{x}_i^n
    \, , \qquad
    \tilde{t}'_n = \frac{(-1)^{n-1}}{n} \sum_{i=1}^\infty \mathsf{y}_i^n
    \, .
    \label{eq:Miwa_var_mod}
\end{align}
In our previous works \cite{Kimura:2020sud, kimura2021unitary}, we considered the $p$-even higher asymptotic analysis of the random partitions and related unitary matrix models in order to study the phase structure of the corresponding physical models of interest. 
This analysis is based on the fact that the scaling limit considered in~\S\ref{sec:scaling} for even $p$ corresponds to the edge scaling limit, namely the large $N$ limit of the unitary matrix model.
For the $p$-odd case, on the other hand, the relation between the cusp scaling limit and the large $N$ limit is not clear at this moment.
We remark that it has been reported that the phase transition does not happen for the unitary matrix model in the specific scaling limit, which may correspond to the cusp limit~\cite{Hisakado:1996di,Hisakado:1996zg,Hisakado:1997ef}.

\section{Level spacing distribution: odd $p$}\label{sec:level_spacing}

Based on the general formalism of the gap probability discussed in \S\ref{sec:gap_prob}, which is for the determinantal point process defined by the limiting kernel, we consider the specific case with the interval
\begin{align}
    I = [-s,s]
    \, , \qquad 
    (a_1,a_2) = (-s,s)
    \, ,
    \label{eq:spacing_interval}
\end{align}
from which we obtain the level spacing distribution by taking the second derivative.
We now focus on the case for odd $p$.
In this case, we have the following auxiliary functions,
\begin{align}
    (\sfq_k, \sfp_k) := (\sfq_{k,2}, \sfp_{k,2}) = (-1)^k (\sfq_{k,1}, -\sfp_{k,1})
    \, .
    \label{eq:qp_parity}
\end{align}
Namely, $(\sfq_{2k},\sfp_{2k-1})$ are even, $(\sfq_{2k-1},\sfp_{2k})$ are odd functions.
\begin{lemma}\label{lem:uv_der}
Denoting the $s$-derivative of a function $f$ by
\begin{align}
 \dv{f}{s} = f'
 \, ,
\end{align}
we have the following relations for the auxiliary functions,
\begin{subequations}
\begin{gather}
    u_{2k-1}' = 2 \sfq_{2k-1} \sfp
    \, , \qquad
    v_{2k-1}' = 2 \sfq \sfp_{2k-1}
    \, , \label{eq:uv_der} \\
    u_{2k} = v_{2k} = 0
    \, .
 \label{eq:uv_even_zero}
\end{gather}
\end{subequations}
\end{lemma}
\begin{proof}
From \eqref{eq:uv_dependence} and \eqref{eq:qp_parity}, the total variations of the auxiliary functions are given by
\begin{subequations}
\begin{align}
    \dd u_\ell %= \sum_{j=1}^{2m} (-1)^j \sfq_{k,j} \sfp_{0,j} \dd a_j 
    & = - \sfq_{\ell,1} \sfp_{0,1} \dd {(-s)} + \sfq_{\ell,2} \sfp_{0,2} \dd{s}
    =
    \begin{cases}
    2 \sfq_\ell \sfp \dd{s} & (\ell:~\text{odd}) \\
    0 & (\ell:~\text{even})
    \end{cases} \\[.5em]
    \dd v_\ell %= \sum_{j=1}^{2m} (-1)^j \sfq_{k,j} \sfp_{0,j} \dd a_j 
    & = - \sfq_{0,1} \sfp_{\ell,1} \dd {(-s)} + \sfq_{0,2} \sfp_{\ell,2} \dd{s}
    =
    \begin{cases}
    2 \sfq \sfp_\ell \dd{s} & (\ell:~\text{odd}) \\
    0 & (\ell:~\text{even})
    \end{cases}
\end{align}
\end{subequations}
Recalling $u_\ell$, $v_\ell \xrightarrow{s \to 0} 0$, we conclude that $u_\ell$, $v_\ell = 0$ for even $\ell$.
\end{proof}

%\paragraph{Derivatives of $(\sfq_k,\sfp_k)$}

We then consider the $s$-dependence of the auxiliary wave functions.
The total derivatives of the auxiliary wave functions are given in \eqref{eq:qp_dependence}.
However, as the interval parameters are taken as in \eqref{eq:spacing_interval}, we have to take into account both the contributions of $a_{j=1,2}$. 
\begin{lemma}\label{lem:s-derivative_qp}
The $s$-derivative of the auxiliary wave functions is given as follows,
\begin{subequations}
\begin{align}
    \sfq_k' & = \sfq_{k+1} - u_k \sfq + 2 (-1)^k R \sfq_k \, , \\
    \sfp_k' & = \sfp_{k+1} - v_k \sfp - 2 (-1)^k R \sfp_k \, ,
\end{align}
\end{subequations}
where we define
\begin{align}
 R & := R_{12} = R_{21}
% \nonumber \\ &
 = \frac{1}{2 s} \sum_{k = 1}^p \sfq_{p - k} \sfp_{k - 1}
 \, .
 \label{eq:R12}
\end{align}
\end{lemma}
\begin{proof}
The total variation of $\sfq_{k}(x)$ is now given by
\begin{align}
    \dd \sfq_{k}(x) & = \qty(\sfq_{k+1}(x) - \sfq(x) u_k + R(x,a_1) \sfq_{k,1} - R(x,a_2) \sfq_{k,2} ) \dd{x}
%    \nonumber \\ & \qquad 
    - R(x,a_1) \sfq_{k,1} \dd{a_1}
    + R(x,a_2) \sfq_{k,2} \dd{a_2}
    \nonumber \\
    & \xrightarrow[\eqref{eq:spacing_interval}]{x \to s}
    \qty( \sfq_{k+1} - \sfq u_k + 2 (-1)^k R \sfq_k
    ) \dd{s}
    \, .
\end{align}
A similar expression holds for $\sfp_k$, and hence we obtain the formulas above.
\end{proof}
\begin{corollary}\label{cor:s-derivative_qp}
The $s$-derivative of $(\sfq_{p-1},\sfp_{p-1})$ is given as follows,
\begin{subequations}
\begin{align}
    \sfq_{p-1}' & = + s \sfq - \sum_{m=1}^{\lfloor p/2 \rfloor} \sfq_{p-2m} v_{2m-1} + 2 R \sfq_{p-1} \, , \\
    \sfp_{p-1}' & = - s \sfp - \sum_{m=1}^{\lfloor p/2 \rfloor} \sfp_{p-2m} u_{2m-1} - 2 R \sfp_{p-1} \, .
\end{align}
\end{subequations}
\end{corollary}
\begin{proof}
This follows from Lemma~\ref{lem:s-derivative_qp} together with the relations \eqref{eq:n_lower} and \eqref{eq:uv_even_zero}.
\end{proof}

%\rem{Lax matrix for $(\mathsf{Q},\mathsf{P})$}

\begin{lemma}\label{lem:sR_derivative}
The following relation holds,
\begin{align}
 \qty(s R)' 
 %& = \sum_{k=1}^{p-1} \sfq_{p-k} \sfp_{k} - \sum_{m=1}^{\lfloor p/2 \rfloor} \qty( u_{2m-1} \sfq \sfp_{p-2m} + v_{2m-1} \sfq_{p-2m} \sfp ) \nonumber \\
 & = \sum_{k=1}^{p-1} \sfq_{p-k} \sfp_{k} - \frac{1}{2} \sum_{m=1}^{\lfloor p/2 \rfloor} \qty( u_{2m-1} v'_{p-2m} + v_{2m-1} u'_{p-2m} )
 \, .
\end{align}
\end{lemma}
\begin{proof}
This can be shown as follows,
\begin{align}
 \qty(2 s R)'
 & = \sum_{k = 1}^p \qty( \sfq_{p - k} \sfp_{k - 1} )'
 \nonumber \\
 & = \sum_{k=1}^p \qty( \sfq_{p-k+1} \sfp_{k-1} + \sfq_{p-k} \sfp_{k} - \qty( u_{p-k} \sfq \sfp_{k-1} + v_{k-1} \sfq_{p-k} \sfp ) )
 \nonumber \\
 & = \sfq \sfp_p + \sfq_p \sfp + 2 \sum_{k=1}^{p-1} \sfq_{p-k} \sfp_{k} - \sum_{k=1}^p \qty( u_{k-1} \sfq \sfp_{p-k} + v_{k-1} \sfq_{p-k} \sfp )
 \nonumber \\
 & \stackrel{\eqref{eq:n_lower}}{=}
 2 \sum_{k=1}^{p-1} \sfq_{p-k} \sfp_{k} - \sum_{k=1}^p (1 + (-1)^k)\qty( u_{k-1} \sfq \sfp_{p-k} + v_{k-1} \sfq_{p-k} \sfp )
 \, .
\end{align}
Applying \eqref{eq:uv_der}, we obtain the relation above.
\end{proof}

\paragraph{Integral of motion}

The integral of motion shown in Lemma~\ref{lem:IOM} is given as follows for the current case,
\begin{align}
    \mathsf{I}_\ell = u_\ell - (-1)^\ell v_\ell + \sum_{k = 1}^\ell (-1)^k \qty( u_{\ell - k} v_{k-1} + (1 - (-1)^\ell) \sfq_{\ell-k} \sfp_{k-1} )
    \, .
\end{align}
For even $\ell$, this becomes trivial.
For odd $\ell$, on the other hand, we have
\begin{align}
   \mathsf{I}_\ell & = u_\ell + v_\ell + \sum_{k = 1}^\ell (-1)^k \qty( u_{\ell - k} v_{k-1} + 2 \sfq_{\ell-k} \sfp_{k-1} ) 
   \nonumber \\
   & = u_\ell + v_\ell + \sum_{m = 1}^{\lfloor \ell/2 \rfloor} u_{\ell - 2m} v_{2m-1} + 2 \sum_{k = 1}^\ell (-1)^k \sfq_{\ell-k} \sfp_{k-1}
   \, .
\end{align}
The lower order examples are given as follows,
\begin{subequations}
\begin{align}
    \mathsf{I}_1 & = u_1 + v_1 - 2 \sfq \sfp 
    \label{eq:IOM1}
    \, , \\
    \mathsf{I}_3 & = u_3 + v_3 + u_1 v_1 - 2 \qty(\sfq_2 \sfp - \sfq_1 \sfp_1 + \sfq \sfp_2)
    \, , \\
    \mathsf{I}_5 & = u_5 + v_5 + u_3 v_1 + u_1 v_3 - 2 \qty( \sfq_4 \sfp - \sfq_3 \sfp_1 + \sfq_2 \sfp_2 - \sfq_1 \sfp_3 + \sfq \sfp_4)
    \, .
\end{align}
\end{subequations}
For $\ell = p$, we obtain the following from \eqref{eq:IOM_highest} (also from Lemma~\ref{lem:pq_rel}),
\begin{align}
    \mathsf{I}_p & = 2 \sum_{k=1}^p (-1)^k \sfq_{p-k} \sfp_{k-1}
    \, .
\end{align}
Hence, using the first integral of motion~\eqref{eq:IOM1} together with the relation \eqref{eq:uv_even_zero}, we may rewrite the function $R$ \eqref{eq:R12} in terms of $(u_{2m-1},v_{2m-1})_{m=1,\ldots,\lfloor p/2 \rfloor}$,
\begin{align}
    R & % = \frac{1}{2s} \sum_{k=1}^p \qty(1 + (-1)^k) \sfq_{p-k} \sfp_{k-1} \nonumber \\
    = \frac{1}{s} \sum_{m=1}^{\lfloor p/2 \rfloor} \sfq_{p-2m} \sfp_{2m-1}
    %= \frac{-1}{s} \sum_{m=0}^{\lfloor p/2 \rfloor} \sfq_{p-2m-1} \sfp_{2m} \nonumber \\ &
    = \frac{1}{2s(u_1 + v_1)} \sum_{m=1}^{\lfloor p/2 \rfloor} u_{p-2m}' v_{2m-1}'
    \, .
    \label{eq:R_uv}
\end{align}

\subsection{Hamiltonian structure}

For the current case, we have the Hamiltonian \eqref{eq:Hj} as follows,
\begin{align}
 H & := H_1 = H_2
 \nonumber \\
 & = - s \sfq \sfp - \sum_{k = 1}^{p-1} (-1)^k \sfq_{p-k} \sfp_{k}
   - \sum_{k=1}^p (-1)^k \qty( u_{k-1} \sfq \sfp_{p-k} - v_{k-1} \sfq_{p-k} \sfp ) - 2 s R^2
   \nonumber \\
   & = - s \sfq \sfp - \sum_{k = 1}^{p-1} (-1)^k \sfq_{p-k} \sfp_{k}
   - \sum_{m=1}^{\lfloor p/2 \rfloor} \qty( u_{2m-1} \sfq \sfp_{p-2m} - v_{2m-1} \sfq_{p-2m} \sfp ) - 2 s R^2
 \, .
\end{align}
In this case, however, we cannot apply the previous expression \eqref{eq:H_dependence} for the $s$-derivative of the Hamiltonian since the derivatives of the auxiliary wave functions are modified as in Lemma~\ref{lem:s-derivative_qp} and Corollary~\ref{cor:s-derivative_qp}.
\begin{lemma}\label{lem:Hamiltonian_derivative}
The $s$-derivative of the Hamiltonian is given as follows,
\begin{align}
    H' & = - \sfq \sfp + 2R^2
    \nonumber \\
    & = - \frac{u_1 + v_1}{2} + \frac{1}{2s^2(u_1 + v_1)^2} \qty( \sum_{m=1}^{\lfloor p/2 \rfloor} u_{p-2m}' v_{2m-1}' )^2
    \, .
\end{align}
\end{lemma}
\begin{proof}
We can show this similarly to Lemma~\ref{lem:Ham_derivative}.
In this case, the $s$-derivative of the operator $\hat{K}$ is obtained from \eqref{eq:K_derivative_aj} as follows,
\begin{align}
    \pdv{}{s} \hat{K} = K \ket{-s} \bra{-s} + K \ket{s} \bra{s}
    \, .
\end{align}
Hence, we obtain
\begin{align}
    \dv{H}{s}
    & = \dv{}{s} \qty( \bra{s} \frac{\hat{K}}{1 - \hat{K}} \ket{s} )
    \nonumber \\
%    & = \bra{a_j} (1 - \hat{K})^{-1} \comm{P}{\hat{K}} (1 - \hat{K})^{-1} + (1 - \hat{K})^{-1} \qty( \pdv{\hat{K}}{a_j} ) (1 - \hat{K})^{-1} \ket{a_j} \nonumber \\
    & = \bra{s} (1 - \hat{K})^{-1} \Big( - \ket{\phi}\bra{\psi} \Pi_I + 2 K \ket{-s}\bra{-s} \Big) (1 - \hat{K})^{-1} \ket{s}
    \nonumber \\
    & = - \sfq \sfp + 2R^2 
    \, .
\end{align}
The expression in terms of $(u_{2m-1},v_{2m-1})_{m=1,\ldots,\lfloor p/2 \rfloor}$ is then obtained from the integral of motion $\mathsf{I}_1$ \eqref{eq:IOM1} and the expression \eqref{eq:R_uv}.
\end{proof}

\begin{remark}[Hamiltonian structure]
We may write the Hamiltonian as follows,
\begin{align}
    H 
    = \sum_{k=1}^p (-1)^k \dv{\sfq_{p-k}}{s} \sfp_{k-1}
    = \sum_{k=1}^p (-1)^{k-1} \sfq_{p-k} \dv{\sfp_{k-1}}{s}
    \, ,
\end{align}
which comes from \eqref{eq:Hamiltonian_Rjj}.
Hence, even though the $s$-derivative of the auxiliary functions are modified in this case, the same Hamilton equations are still available as in the form of \eqref{eq:Hamilton_eq}.
\end{remark}

\paragraph{Fredholm determinant}

In this case, the logarithmic derivative of the Fredholm determinant \eqref{eq:logF_dependence} is given by
\begin{align}
    \dd \log F(s) & = H_1 \dd{(-s)} - H_2 \dd{s}  = - 2 H \dd{s}
    \, .
\end{align}
Denoting the Hamiltonian as a function of the $s$-variable by $H(s)$ and noticing $F(s) \xrightarrow{s \to 0} 1$, we obtain the Fredholm determinant in terms of the Hamiltonian as follows,
\begin{align}
 F(s) & = \exp \qty( - 2 \int_0^s \dd{\sigma} H(\sigma) )
 \nonumber \\
 & = \exp \qty( - 2 \int_0^s \dd{\sigma} (\sigma - s) \qty( \sfq \sfp - 2 R^2 ) )
 \, .
\end{align}
The second expression is analogous to the Tracy--Widom distribution in terms of the Hastings–-McLeod solution to the Painlevé II equation~\cite{Hastings:1980ARMA} and its higher analogs.

\subsection{Nonlinear differential equations}\label{sec:nlin_eqs}

We derive the closed differential equations for $(\underline{u},\underline{v})$ by using the functional relations obtained above.
The strategy that we apply is to rewrite all the bilinear forms $(\sfq_k \sfp_l)_{k,l=0,\ldots,p-1}$ in terms of $(\underline{u},\underline{v})$ recursively.

\begin{proposition}
All the bilinear terms $(\sfq_k \sfp_l)_{k,l=0,\ldots,p-1}$ are expressed in terms of the auxiliary functions $(\underline{u},\underline{v})$.
\end{proposition}

\begin{proof}
First of all, the lowest bilinear term is obtained from the first integral of motion $\mathsf{I}_p$ \eqref{eq:IOM1} as follows,
\begin{align}
    \sfq \sfp = \frac{u_1 + v_1}{2}
    \, .
\end{align}
From Lemma~\ref{lem:uv_der}, we then obtain
\begin{align}
    \sfq_{2m-1} \sfp = \frac{1}{2} u_{2m-1}'
    \, , \qquad
    \sfq \sfp_{2m-1} = \frac{1}{2} v_{2m-1}'
    \, ,
    \label{eq:qp_0o}
\end{align}
and 
\begin{align}
    \sfq_{2m-1} \sfp_{2n-1} = \frac{u_{2m-1}' v_{2n-1}'}{2(u_1 + v_1)}
    \, .
    \label{eq:qp_oo}
\end{align}

We then apply the recursion relations.
Assume that the bilinear term $\sfq_k \sfp_l$ is known: We can write it in terms of $(\underline{u},\underline{v})$.
From Lemma~\ref{lem:s-derivative_qp} and Corollary~\ref{cor:s-derivative_qp}, the $s$-derivative of $\sfq_k \sfp_l$ gives rise to two possibly unknown terms, $\sfq_{k+1} \sfp_{l}$ and $\sfq_k \sfp_{l+1}$, 
\begin{align}
    (\sfq_k \sfp_l)' = \sfq_{k+1} \sfp_{l} + \sfq_k \sfp_{l+1} + \cdots
    \, ,
\end{align}
where the symbol $\cdots$ means the known terms.
We remark that the expression of the function $R$ in terms of $(\underline{u},\underline{v})$ is also known as in \eqref{eq:R_uv} as a consequence of the highest integral of motion $\mathsf{I}_p$ \eqref{eq:IOM_highest}.
If $\sfq_k \sfp_{l+1}$ is already known, we may write
\begin{align}
    \sfq_{k+1} \sfp_{l} = (\sfq_k \sfp_l)' + \cdots
    \, .
\end{align}
If $\sfq_{k+1} \sfp_{l}$ is known, we can similarly determine the other bilinear term $\sfq_k \sfp_{l+1}$.
When $(k,l) = (\text{odd},\text{odd})$ or $(\text{even},\text{even})$, we cannot determine them in this way since both $\sfq_{k+1} \sfp_{l}$ and $\sfq_k \sfp_{l+1}$ are unknown.
We should fix them from another path.

We start this algorithm from $\sfq_{2m-1} \sfp$ and $\sfq \sfp_{2m-1}$, whose expressions are known as \eqref{eq:qp_0o}:
\begin{subequations}
\begin{equation}
    \begin{tikzcd}
    \sfq_{2m-1} \sfp \arrow[r] & \sfq_{2m} \sfp \arrow[r] & \sfq_{2m} \sfp_1 \arrow[r] & \sfq_{2m} \sfp_2 \arrow[r] & \cdots \arrow[r] & \sfq_{2m} \sfp_{2m}
    \end{tikzcd}
\end{equation}
\begin{equation}
    \begin{tikzcd}
    \sfq \sfp_{2m-1} \arrow[r] & \sfq \sfp_{2m} \arrow[r] & \sfq_1 \sfp_{2m} \arrow[r] & \sfq_{2} \sfp_{2m} \arrow[r] & \cdots \arrow[r] & \sfq_{2m} \sfp_{2m}
    \end{tikzcd}    
\end{equation}
\end{subequations}
Meanwhile, we have $\sfq_{2m} \sfp_{2n}$ and $\sfq_{2n} \sfp_{2m}$ for $n \le m$.
In this case, we can fix $\sfq_{2m+1} \sfp_{2n}$ and $\sfq_{2n} \sfp_{2m+1}$ from $\sfq_{2m+1} \sfp_{2n-1}$ and $\sfq_{2n-1} \sfp_{2m+1}$, respectively.
\end{proof}
We show the flow diagram to determine the bilinear terms $(\sfq_k \sfp_l)_{0 \le k,l \le p-1}$ for $p = 7$:
\begin{equation}
\begin{tikzcd}
 \sfq_{} \sfp_{} & \sfq_{} \sfp_{1} \arrow[r] & \sfq_{} \sfp_{2} \arrow[d] & \sfq_{} \sfp_{3} \arrow[r] & \sfq_{} \sfp_{4} \arrow[d] & \sfq_{} \sfp_{5} \arrow[r] & \sfq_{} \sfp_{6} \arrow[d] \\
 \sfq_{1} \sfp_{} \arrow[d] & \sfq_{1} \sfp_{1} & \sfq_{1} \sfp_{2} \arrow[d] & \sfq_{1} \sfp_{3} \arrow[d] & \sfq_{1} \sfp_{4} \arrow[d] & \sfq_{1} \sfp_{5} \arrow[d] & \sfq_{1} \sfp_{6} \arrow[d] \\
 \sfq_{2} \sfp_{2} \arrow[r] & \sfq_{2} \sfp_{1} \arrow[r] & \sfq_{2} \sfp_{2} & \sfq_{2} \sfp_{3} & \sfq_{2} \sfp_{4} \arrow[d] & \sfq_{2} \sfp_{5} & \sfq_{2} \sfp_{6} \arrow[d] \\
 \sfq_{3} \sfp_{} \arrow[d] & \sfq_{3} \sfp_{1} \arrow[r] & \sfq_{3} \sfp_{2} & \sfq_{3} \sfp_{3} & \sfq_{3} \sfp_{4} \arrow[d] & \sfq_{3} \sfp_{5} \arrow[d] & \sfq_{3} \sfp_{6} \arrow[d] \\
 \sfq_{4} \sfp_{} \arrow[r] & \sfq_{4} \sfp_{1} \arrow[r] & \sfq_{4} \sfp_{2} \arrow[r] & \sfq_{4} \sfp_{3} \arrow[r] & \sfq_{4} \sfp_{4} & \sfq_{4} \sfp_{5} & \sfq_{4} \sfp_{6} \arrow[d] \\
 \sfq_{5} \sfp_{} \arrow[d] & \sfq_{5} \sfp_{1} \arrow[r] & \sfq_{5} \sfp_{2} & \sfq_{5} \sfp_{3} \arrow[r] & \sfq_{5} \sfp_{4} & \sfq_{5} \sfp_{5} & \sfq_{5} \sfp_{6} \arrow[d] \\
 \sfq_{6} \sfp_{} \arrow[r] & \sfq_{6} \sfp_{1} \arrow[r] & \sfq_{6} \sfp_{2} \arrow[r] & \sfq_{6} \sfp_{3} \arrow[r] & \sfq_{6} \sfp_{4} \arrow[r] & \sfq_{6} \sfp_{5} \arrow[r] & \sfq_{6} \sfp_{6} \\
\end{tikzcd}
\end{equation}
Based on these expressions, we obtain the closed differential equations for the auxiliary functions $(\underline{u},\underline{v})$ as follows.
\begin{proposition}\label{prop:nlin_eqs}
The auxiliary functions $(\underline{u},\underline{v}) = (u_{2m-1},v_{2m-1})_{m=1,\ldots,\lfloor p/2 \rfloor}$ obey the closed coupled nonlinear differential equations.
The number of the differential equations agrees with the number of the auxiliary functions, $2 \lfloor p/2 \rfloor = p-1$.
\end{proposition}
\begin{proof}
For $m = 1,\ldots,\lfloor p/2 \rfloor$, the bilinear term $\sfq_{2m} \sfp_{2m}$ has two different forms,
\begin{align}
    \sfq_{2m} \sfp_{2m} = (\sfq_{2m-1} \sfp_{2m})' + \text{\color{c1} known}
    \, , \qquad
    \sfq_{2m} \sfp_{2m} = (\sfq_{2m-1} \sfp_{2m-1})' + \text{\color{magenta} known}
    \, .
\end{align}
By equating these expressions, we obtain the following equation
\begin{align}
    \sfq_{2m} \sfp_{2m} : \quad 
    (\sfq_{2m-1} \sfp_{2m})' + \text{\color{c1} known} =  (\sfq_{2m-1} \sfp_{2m-1})' + \text{\color{magenta} known}
    \, .
\end{align}
For $\sfq_{p-1} \sfp_{p-1}$, in particular, we have another equation from its derivative,
\begin{align}
    (\sfq_{p-1} \sfp_{p-1})' : \quad
    &
    \frac{1}{2} \qty[
    (\sfq_{p-2} \sfp_{p-1})' + \text{\color{c1} known}
    + (\sfq_{p-1} \sfp_{p-2})' + \text{\color{magenta} known} ]'
    \nonumber \\
    & = - s (\sfq_{p-1} \sfp - \sfq \sfp_{p-1}) - \sum_{m=1}^{\lfloor p/2 \rfloor} (u_{2m-1} \sfq_{p-1} \sfp_{p-2m} + v_{2m-1} \sfq_{p-2m} \sfp_{p-1})
    \, ,
\end{align}
where the RHS is obtained from Corollary~\ref{cor:s-derivative_qp}.
In addition, the remaining integrals of motion provide the equations, $\mathsf{I}_{2m-1} = 0$ for $m = 2,\ldots,\lfloor p/2 \rfloor$.
Writing all the bilinear terms $(\sfq_k \sfp_l)_{0 \le k,l \le p-1}$ in terms of the auxiliary functions $(\underline{u},\underline{v})$, we obtain differential equations for them.
The total number of the equations is given by $\lfloor p/2 \rfloor + 1 + (\lfloor p/2 \rfloor - 1) = 2 \lfloor p/2 \rfloor = p-1$.
\end{proof}

\subsection{Large gap behavior $s \to \infty$}\label{sec:large_gap}

As shown before, the logarithmic derivative of the Fredholm determinant is given by the Hamiltonian.
Moreover, the $s$-derivative of the Hamiltonian is written in terms of the auxiliary functions $(u_{2m-1},v_{2m-1})_{m=1,\ldots,\lfloor p/2 \rfloor}$ as shown in Lemma~\ref{lem:Hamiltonian_derivative}.
In particular, using the integral of motion \eqref{eq:IOM1}, we obtain
\begin{align}
 \dv[2]{}{s} \log F(s) & = - 2H' = u_1 + v_1 - 4 R^2
 \, .
\end{align}
Hence, the large gap behavior of the Fredholm determinant is obtained from the asymptotic behavior of these auxiliary functions.

\begin{lemma}\label{lem:H_der_asymp}
In the large $s$ limit, the $s$-derivative of the Hamiltonian behaves as follows,
\begin{align}
    H' \ \xrightarrow{s \to \infty} \ O(s^{\frac{2}{p}})
    \, .
    \label{eq:H_der_asymp}
\end{align}
\end{lemma}

\begin{proof}
We first consider the large $s$ behavior of the function $R$, 
\begin{align}
    R = R(-s,s) = K(-s,s) + \sum_{n=1}^\infty \int_I \dd{z}_1 \cdots \int_I \dd{z}_n K(-s,z_1) K(z_1,z_2) \cdots K(z_{n-1},z_n) K(z_n,s)
    \, .
    \label{eq:R_sum}
\end{align}
From the Cristoffel--Darboux formula \eqref{eq:CD_formula} together with the asymptotic behavior of the higher Airy functions, \eqref{eq:Ai_asymp1} and \eqref{eq:Ai_asymp2}, we have the following asymptotic behavior of the kernel,
\begin{subequations}\label{eq:kernel_asymp}
\begin{align}
    K(s,s) \ \xrightarrow{s \to \infty} \ & O(s^{\frac{1}{p}}) \, , \\
    K(-s,s) \ \xrightarrow{s \to \infty} \ & O(s^{-1}) \, .
\end{align}
\end{subequations}
Hence, we obtain
\begin{align}
    \int_I \dd{z} K(-s,z) K(z,s) \ \xrightarrow{s \to \infty} \ O(s^{\frac{1}{p}})
    \, .
\end{align}
Recalling the projectivity property of the kernel \eqref{eq:projectivity_Airy},
\begin{align}
    \lim_{s \to \infty} \int_I \dd{z} K(x,z) K(z,y) = K(x,y)
    \, ,
\end{align}
we see that the $n > 1$ terms in \eqref{eq:R_sum} provide the same order contribution, from which we obtain
\begin{align}
    R \ \xrightarrow{s \to \infty} \ O(s^{\frac{1}{p}})
    \, .
\end{align}

We then consider the following term,
\begin{align}
    u_1 + v_1 & = \int_I \dd{x} \qty( \psi(x) \phi'(x) + \psi'(x) \phi(x) ) + \int_I \dd{x} \int_I \dd{y} \qty( \psi(x) R(x,y) \phi'(x) + \psi'(x) R(x,y) \phi(y) )
    \nonumber \\
    & = 2 \psi(s) \phi(s) + \int_I \dd{x} \int_I \dd{y} \qty( \psi(x) K(x,y) \phi'(x) + \psi'(x) K(x,y) \phi(y) ) + \cdots
    \, .
\end{align}
The asymptotic behavior of the first term is immediately obtained from \eqref{eq:Ai_asymp1} and \eqref{eq:Ai_asymp2},
\begin{align}
    \psi(s) \phi(s) \ \xrightarrow{s \to \infty} \ O(s^{-1+\frac{1}{p}})
    \, .
\end{align}
The integral term behaves as follows,
\begin{align}
    & \int_I \dd{x} \int_I \dd{y} \qty( \psi(x) K(x,y) \phi'(x) + \psi'(x) K(x,y) \phi(y) )
    \nonumber \\
    & = \int_I \dd{x} \int_I \dd{y} \qty[ \qty( \pdv{}{x} + \pdv{}{y} ) \psi(x) K(x,y) \phi(y)
    - \psi(x) \qty( \pdv{}{x'} + \pdv{}{y'} ) K(x',y')\Big|_{(x',y') \to (x,y)} \psi(y)
    ] \nonumber \\
    & \stackrel{\eqref{eq:P_K_commutator}}{=}
    \qty[\int_I \dd{y} \psi(x) K(x,y) \phi(y)]^{x=+s}_{x=-s} + \qty[\int_I \dd{y} \psi(x) K(x,y) \phi(y)]^{y=+s}_{y=-s}
    + \qty( \int_I \dd{x} \psi(x) \phi(x) )^2
    \nonumber \\
    & = 2 \psi(s) \qty( \int_I \dd{y} K(s,y) \phi(y) ) + 2 \qty( \int_I \dd{x} \psi(x) K(x,s) ) \psi(s)
    % \nonumber \\ &
    \ \xrightarrow[\eqref{eq:kernel_asymp}]{s \to \infty} \ O(s^{\frac{2}{p}})
    \, .
\end{align}
From the projectivity of the kernel, the multiple integral terms provide the same order contributions as before.

Summarizing all these contributions, we conclude the asymptotic behavior of $H'$ as shown in \eqref{eq:H_der_asymp}.
\end{proof}

From Lemma~\ref{lem:H_der_asymp}, we immediately obtain the large gap behavior of the Fredholm determinant as follows.
\begin{proposition}\label{prop:large_gap}
The Fredholm determinant behaves as
\begin{align}\label{eq:large_gap}
    F(s) \ \xrightarrow{s \to \infty} \ \exp \qty( - C_p s^{\frac{2}{p}+2} )
\end{align}
with a $p$-dependent positive constant $C_p$.
\end{proposition}
\begin{remark}
This result is consistent with the Forrester--Chen--Eriksen--Tracy conjecture~\cite{Forrester:1993vtx,Chen:1995uy} on the large gap asymptotics of the Fredholm determinant $F(s) \xrightarrow{s \to \infty} \exp \qty(- C s^{2\beta + 2})$ for the density of state behaving as $\rho(s) \sim s^\beta$ ($\beta = \frac{1}{p}$ in our case).
\end{remark}

Based on the current framework basead on the saddle point analysis, it seems to be difficult to provide a rigorous argument and to determine the constant $C_p$ in general.
It would be plausible to apply the Riemann--Hilbert analysis to this issue as discussed in~\cite{Bleher:2006CMP} for the case $p = 3$.
See~\cite{dai2021asymptotics} for detailed analysis on the gap probability asymptotics.

\subsection{Example: $p=3$}

Let us consider the simplest example $p=3$ for more detail.
We simply write $(u_1,v_1) = (u,v)$.
The $s$-derivatives of the auxiliary functions are given as follows,\\
\begin{subequations}
\begin{minipage}{.5\textwidth}
\begin{align}
    \sfq' & = \sfq_1 + 2 R \sfq 
    \, , \\
    \sfq_1' & = \sfq_2 - u \sfq - 2 R \sfq_1
    \, , \\
    \sfq_2' & = s \sfq - v \sfq_1 + 2 R \sfq_2
    \, ,
\end{align}
\end{minipage}
\begin{minipage}{.5\textwidth}
\begin{align}
    \sfp' & = \sfp_1 - 2 R \sfp 
    \, , \\
    \sfp_1' & = \sfp_2 - v \sfp + 2 R \sfp_1
    \, , \\
    \sfp_2' & = - s \sfp - u \sfp_1 - 2 R \sfp_2
    \, ,
\end{align}
\end{minipage}
\begin{align}
    u' = 2 \sfq_1 \sfp 
    \, , \qquad
    v' = 2 \sfq \sfp_1
    \, .
\end{align}
\end{subequations}
Hence, we have the Lax matrix form,
\begin{align}
    \dv{}{s}
    \begin{pmatrix}
    \sfq \\ \sfq_1 \\ \sfq_2
    \end{pmatrix}
    =
    \begin{pmatrix}
    2 R & 1 & 0 \\
    - u & - 2R & 1 \\
    s & - v & 2 R
    \end{pmatrix}
    \begin{pmatrix}
    \sfq \\ \sfq_1 \\ \sfq_2
    \end{pmatrix}
    \, , \qquad
    \dv{}{s}
    \begin{pmatrix}
    \sfp \\ \sfp_1 \\ \sfp_2
    \end{pmatrix}
    =
    \begin{pmatrix}
    - 2 R & 1 & 0 \\
    - v & 2R & 1 \\
    - s & - u & - 2 R
    \end{pmatrix}
    \begin{pmatrix}
    \sfp \\ \sfp_1 \\ \sfp_2
    \end{pmatrix}    
\end{align}
The integrals of motion are given by
\begin{subequations}
\begin{align}
    \mathsf{I}_1 & = u + v - 2 \sfq \sfp \, , \\
    \mathsf{I}_3 & = \sfq_2 \sfp - \sfq_1 \sfp_1 + \sfq \sfp_2 \, .
\end{align}
\end{subequations}
The Hamiltonian and its $s$-derivative are given by
\begin{subequations}
\begin{align}
    H & = - s \sfq \sfp + \sfq_2 \sfp_1 - \sfq_1 \sfp_2 - u \sfq \sfp_1 + v \sfq_1 \sfp - 2 s R^2
    \, ,\\
    H' & = - \sfq \sfp + 2 R^2 = - \frac{u+v}{2} + \frac{1}{2} \qty( \frac{u' v'}{s(u+v)} )^2
    \, ,
\end{align}
\end{subequations}
where the function $R$ is given by
\begin{align}
    R = \frac{u' v'}{2s(u+v)}
    \, .
    \label{eq:R_uv_p=3}
\end{align}
From Lemma~\ref{lem:sR_derivative}, we obtain
\begin{align}
    (2 s R)' = 2 \sfq_2 \sfp_1 + 2 \sfq_1 \sfp_2 - u v' - u' v
    \, .
\end{align}

\paragraph{Nonlinear differential equations}
We consider the derivatives of the auxiliary functions $(u,v)$ recursively:
\begin{subequations}
\begin{align}
    u'' & = 2 \sfq_2 \sfp - u (u+v) - 4 R u' + 2 s R
   &  \implies  \quad
    2 \sfq_2 \sfp & = u'' + u(u+v) + 4 R u' - 2 s R
    \\
    v'' & = 2 \sfq \sfp_2 - v (u+v) - 4 R v' + 2 s R
    & \implies  \quad
    2 \sfq \sfp_2 & = v'' + v(u+v) + 4 R v' - 2 s R
\end{align}
\end{subequations}
which is equivalent to
\begin{subequations}
\begin{align}
    2 \sfq_2 \sfp + 2 \sfq \sfp_2 & = u'' + v'' + (u + v)^2 + 4 R (u' + v') - 4 s R \, , \\
    2 \sfq_2 \sfp - 2 \sfq \sfp_2 & = u'' - v'' + (u^2 - v^2) + 4 R (u' - v')
    \, .
\end{align}
\end{subequations}
Hence, $\sfq_2 \sfp$ and $\sfq \sfp_2$ are written only in terms of $(u,v)$ since the function $R$ is given as \eqref{eq:R_uv_p=3}.
We then obtain
\begin{subequations}
\begin{align}
    (2 \sfq_2 \sfp)' & = +s(u+v) - u' v + 2 \sfq_2 \sfp_1
    &  \implies  \quad
    2 \sfq_2 \sfp_1 & = (2 \sfq_2 \sfp)' - s(u+v) + u' v \\
    (2 \sfq \sfp_2)' & = -s(u+v) - u v' + 2 \sfq_1 \sfp_2
    &  \implies  \quad
    2 \sfq_1 \sfp_2 & = (2 \sfq \sfp_2)' + s(u+v) + u v'
\end{align}
\end{subequations}
which yields the expressions of $\sfq_2 \sfp_1$ and $\sfq_1 \sfp_2$ in terms of $(u,v)$.
We may also write
\begin{subequations}
\begin{align}
    2 \sfq_2 \sfp_1 + 2 \sfq_1 \sfp_2 & = 
    (2 \sfq_2 \sfp)' + (2 \sfq \sfp_2)' + u' v + u v' 
    \, , \\
    2 \sfq_2 \sfp_1 - 2 \sfq_1 \sfp_2 & = (2 \sfq_2 \sfp)' - (2 \sfq \sfp_2)' - 2 s (u + v)
    \, .
\end{align}
\end{subequations}
Similarly, we have
\begin{subequations}
\begin{align}
    (2 \sfq_2 \sfp_1)' & = 2 \sfq_2 \sfp_2 + 8 R \sfq_2 \sfp_1 - 2 v \sfq_2 \sfp - 2 s v R + s v'
    \\
    (2 \sfq_1 \sfp_2)' & = 2 \sfq_2 \sfp_2 - 8 R \sfq_1 \sfp_2 - 2 u \sfq \sfp_2 - 2 s u R - s u'
\end{align}
\end{subequations}
from which we obtain
\begin{align}
    4 \sfq_2 \sfp_2 & = (2 \sfq_2 \sfp_1)' + (2 \sfq_1 \sfp_2)' - 8 R (\sfq_2 \sfp_1 + \sfq_1 \sfp_2) - 2 (u \sfq \sfp_2 + v \sfq_2 \sfp) - 2 s (u + v) R - s (u' - v')
    \, .
\end{align}
Therefore, we obtain the coupled differential equations for the auxiliary functions $(u,v)$ from the following relations,%
\footnote{%
We suspect that the differential equations shown in \cite[eqs.(3.25) and (3.26)]{Brezin:1998PREb} need to be improved.}
\begin{subequations}
\begin{align}
    (2 \sfq_2 \sfp_1)' - (2 \sfq_1 \sfp_2)' & = 8 R (\sfq_2 \sfp_1 + \sfq_1 \sfp_2) + 2 (u \sfq \sfp_2 - v \sfq_2 \sfp) - 2 s (u - v) R - s (u' + v') \, ,
    \\
    (\sfq_2 \sfp_2)' & = s (\sfq \sfp_2 - \sfq_2 \sfp) - v \sfq_1 \sfp_2 - u \sfq_2 \sfp_1 \, .
\end{align}
\end{subequations}
Introducing the symmetric and anti-symmetric variables,
\begin{align}
    x := u + v \, , \qquad y := u - v
    \, ,
\end{align}
we obtain
\begin{subequations}
\begin{align}
    y^{(4)} & = 
    -\frac{x'^6}{s^2x^3}-\frac{x'^5}{s^3x^2}+\frac{x'^5}{2sx^3}+\frac{2y'^2x'^4}{s^2x^3}-\frac{2y'x'^4}{sx^3}+\frac{3x''x'^4}{s^2x^2}+\frac{2y'^2x'^3}{s^3x^2}-\frac{y'^2x'^3}{sx^3}+\frac{yx'^3}{sx} \nonumber \\ &
    -\frac{2y'x'^3}{s^2x^2}-\frac{x''x'^3}{sx^2}-\frac{2y'y''x'^3}{s^2x^2}+\frac{2y''x'^3}{sx^2}+\frac{5x'^3}{s}-\frac{y'^4x'^2}{s^2x^3}+\frac{2y'^3x'^2}{sx^3}-\frac{5yx'^2}{8x}-\frac{yy'x'^2}{sx} \nonumber \\ &
    -\frac{y'x'^2}{s}-\frac{2y'x'^2}{s^3x}-\frac{4y'^2x''x'^2}{s^2x^2}+\frac{5y'x''x'^2}{sx^2}+\frac{y'y''x'^2}{sx^2}+\frac{2y''x'^2}{s^2x}+\frac{x^{(3)}x'^2}{sx}-\frac{y^{(3)}x'^2}{sx} \nonumber \\ &
    -\frac{y'^4x'}{s^3x^2}+\frac{y'^4x'}{2sx^3}+\frac{2y'^3x'}{s^2x^2}-\frac{yy'^2x'}{sx}-\frac{5y'^2x'}{s}+\frac{3}{2}sx'-4y'x'+\frac{y'^2x''x'}{sx^2}+\frac{4y'x''x'}{s^2x} \nonumber \\ &
    +\frac{2y'^3y''x'}{s^2x^2}-\frac{6y'^2y''x'}{sx^2}-\frac{4x''y''x'}{sx}-\frac{2y'x^{(3)}x'}{sx}+\frac{yy'^3}{sx}+\frac{y'^3}{s}+\frac{2y'^3}{s^3x}+\frac{5yy'^2}{8x}-\frac{2y'x''^2}{sx} \nonumber \\ &
    +\frac{6y'y''^2}{sx}+2x+\frac{y'^4x''}{s^2x^2}-\frac{y'^3x''}{sx^2}-\frac{y'^3y''}{sx^2}-\frac{6y'^2y''}{s^2x}-4xy''-\frac{y'^2x^{(3)}}{sx}+\frac{3y'^2y^{(3)}}{sx}
    \, ,
    \end{align}
    \begin{align}
    x^{(5)} & =
    \frac{3x'^7}{s^2x^4}+\frac{4x'^6}{s^3x^3}+\frac{9x'^6}{2sx^4}-\frac{6y'^2x'^5}{s^2x^4}-\frac{12x''x'^5}{s^2x^3} \nonumber \\ &
    +\frac{3x'^5}{s^4x^2}+\frac{11x'^5}{2s^2x^3}-\frac{3x'^5}{x^4}-\frac{8y'^2x'^4}{s^3x^3}-\frac{3y'^2x'^4}{sx^4}-\frac{11x''x'^4}{s^3x^2}-\frac{39x''x'^4}{2sx^3} \nonumber \\ &
    +\frac{8y'y''x'^4}{s^2x^3}+\frac{3x^{(3)}x'^4}{s^2x^2}+\frac{2x'^4}{sx}+\frac{6x'^4}{s^3x^2}+\frac{3y'^4x'^3}{s^2x^4}-\frac{6y'^2x'^3}{s^4x^2}-\frac{5y'^2x'^3}{s^2x^3}+\frac{3y'^2x'^3}{x^4} \nonumber \\ &
    +\frac{12x''^2x'^3}{s^2x^2}-\frac{2y''^2x'^3}{s^2x^2}+\frac{4yy'x'^3}{sx^2}+\frac{16y'^2x''x'^3}{s^2x^3}-\frac{20x''x'^3}{s^2x^2}+\frac{9x''x'^3}{x^3}+\frac{8y'y''x'^3}{s^3x^2} \nonumber \\ &
    +\frac{8y'y''x'^3}{sx^3}+\frac{8x^{(3)}x'^3}{sx^2}-\frac{2y'y^{(3)}x'^3}{s^2x^2}+\frac{6x'^3}{s^4x}-\frac{x'^3}{x}-\frac{4x'^3}{s^2}+\frac{4y'^4x'^2}{s^3x^3}-\frac{3y'^4x'^2}{2sx^4} \nonumber \\ &
    -\frac{4y'^2x'^2}{sx}-\frac{6y'^2x'^2}{s^3x^2}+\frac{24x''^2x'^2}{sx^2}-\frac{5y''^2x'^2}{sx^2}+\frac{4yy'x'^2}{s^2x}-\frac{2y'x'^2}{x}+\frac{14y'^2x''x'^2}{s^3x^2} \nonumber \\ &
    +\frac{7y'^2x''x'^2}{sx^3}+\frac{12x''x'^2}{s}-\frac{18x''x'^2}{s^3x}-\frac{8y'^3y''x'^2}{s^2x^3}-\frac{4yy''x'^2}{sx}+\frac{11y'y''x'^2}{s^2x^2}-\frac{6y'y''x'^2}{x^3} \nonumber \\ &
    -\frac{14y'x''y''x'^2}{s^2x^2}-\frac{4y'^2x^{(3)}x'^2}{s^2x^2}+\frac{8x^{(3)}x'^2}{s^2x}-\frac{7x^{(3)}x'^2}{2x^2}-\frac{5y'y^{(3)}x'^2}{sx^2}-\frac{2x^{(4)}x'^2}{sx}+\frac{3y'^4x'}{s^4x^2} \nonumber \\ &
    -\frac{y'^4x'}{2s^2x^3}-\frac{4yy'^3x'}{sx^2}-4x^2x'-4y^2x'-\frac{6y'^2x'}{s^4x}+\frac{y'^2x'}{x}+\frac{4y'^2x'}{s^2}-\frac{8y'^2x''^2x'}{s^2x^2}+\frac{18x''^2x'}{s^2x} \nonumber \\ &
    -\frac{6x''^2x'}{x^2}+\frac{6y'^2y''^2x'}{s^2x^2}-\frac{6y''^2x'}{s^2x}+\frac{3y''^2x'}{x^2}-\frac{4y'^4x''x'}{s^2x^3}+\frac{8y'^2x''x'}{s^2x^2}-\frac{3y'^2x''x'}{x^3} \nonumber \\ &
    -\frac{8yy'x''x'}{sx}-12x''x'-\frac{8y'^3y''x'}{s^3x^2}+\frac{4y'^3y''x'}{sx^3}-\frac{8y'y''x'}{s}+\frac{12y'y''x'}{s^3x}-\frac{14y'x''y''x'}{sx^2} \nonumber \\ &
    -\frac{2y'^2x^{(3)}x'}{sx^2}-\frac{16x''x^{(3)}x'}{sx}+\frac{2y'^3y^{(3)}x'}{s^2x^2}-\frac{6y'y^{(3)}x'}{s^2x}+\frac{3y'y^{(3)}x'}{x^2}+\frac{6y''y^{(3)}x'}{sx}+\frac{x^{(4)}x'}{x} \nonumber \\ &
    +\frac{2y'y^{(4)}x'}{sx}+\frac{2y'^4}{sx}-\frac{4yy'^3}{s^2x}+\frac{2y'^3}{x}-\frac{6x''^3}{sx}-\frac{2y'^2x''^2}{sx^2}-\frac{3y'^2y''^2}{sx^2}+\frac{6x''y''^2}{sx}-8xyy' \nonumber \\ &
    +\frac{y'}{2}-\frac{3y'^4x''}{s^3x^2}+\frac{y'^4x''}{2sx^3}-\frac{4y'^2x''}{s}+\frac{6y'^2x''}{s^3x}+\frac{y'^3y''}{s^2x^2}+\frac{12yy'^2y''}{sx}-\frac{3}{2}sy''+\frac{6y'^3x''y''}{s^2x^2} \nonumber \\ &
    -\frac{12y'x''y''}{s^2x}+\frac{3y'x''y''}{x^2}+\frac{y'^4x^{(3)}}{s^2x^2}-\frac{2y'^2x^{(3)}}{s^2x}+\frac{y'^2x^{(3)}}{2x^2}-6xx^{(3)}+\frac{3x''x^{(3)}}{x}+\frac{4y'y''x^{(3)}}{sx} \nonumber \\ &
    -\frac{y'^3y^{(3)}}{sx^2}-2yy^{(3)}+\frac{6y'x''y^{(3)}}{sx}-\frac{3y''y^{(3)}}{x}-\frac{y'y^{(4)}}{x}
    \, .
\end{align}
\end{subequations}
Although we have such an explicit form of the differential equations, it is not clear how it would be helpful to characterize the level spacing distribution. 
It would be also interesting to compare with the differential equations obtained in \cite{Adler:2012ega}.

\appendix
\section{Single Hamiltonian analysis}\label{sec:single_H}

In this Appendix, we analyze the Hamiltonian system only with a single time variable, which corresponds to the Fredholm determinant with the interval
\begin{align}
    I = [s,\infty)
    \qquad
    \qty(a_1 = s)  
    \, .
\end{align}
For even $p$, the corresponding Fredholm determinant leads to the higher Tracy--Widom distribution, while for odd $p$, this is a formal calculation.
In this case, the derivative of the Fredholm determinant~\eqref{eq:logF_dependence} is given by
\begin{align}
    \dd{\log F(s)} = - H(s) \dd{s}
    \, ,
\end{align}
where the $s$-dependence of the Hamiltonian~\eqref{eq:H_dependence} is given by
\begin{align}
    \dv{H}{s} = - \sfq(s) \sfp(s)
    \, ,
    \label{eq:H_der_TW}
\end{align}
with
\begin{align}
    \sfq(s) = \sfq_0(s)
    \, , \qquad
    \sfp(s) = \sfp_0(s)
    \, .
\end{align}
Therefore, the Fredholm determinant is written as follows,
\begin{align}
    F(s) = \exp \qty( - \int_s^\infty \dd{\sigma} (\sigma - s) \sfq(\sigma) \sfp(\sigma) )
    \, .
\end{align}
For even $p$, since $\sfp(s) = \sfq(s)$, it is given by
\begin{align}
    F(s) = \exp \qty( - \int_s^\infty \dd{\sigma} (\sigma - s) \sfq(\sigma)^2 )
    \, ,
\end{align}
which reproduces the higher Tracy--Widom distribution of even $p$~\cite{Claeys:2009CPAM,Akemann:2012bw,LeDoussal:2018dls,Cafasso:2019IMRN}.

\paragraph{Derivatives of $(\sfq_k,\sfp_k)$}

From the parameter dependence of the wave functions~\eqref{eq:qp_dependence}, the $(k+1)$-st function is written in terms of the lower degree functions,
\begin{subequations}
\begin{align}
    \dv{\sfq_k}{s} & = \sfq_{k+1} - u_k \sfq 
    \quad \implies \quad 
    \sfq_{k+1} = \dv{\sfq_k}{s} + u_k \sfq 
    \, , \\
    \dv{\sfp_k}{s} & = \sfp_{k+1} - v_k \sfp
    \quad \implies \quad 
    \sfp_{k+1} = \dv{\sfp_k}{s} + v_k \sfp
    \, .
\end{align}
\end{subequations}
Similarly, the $s$-dependence of the auxiliary function~\eqref{eq:uv_dependence} is given by
\begin{align}
    \dv{u_k}{s} = - \sfp \sfq_k
    \, , \qquad
    \dv{v_k}{s} = - \sfp_k \sfq
    \, .
\end{align}
Then, the lower degree cases are given as follows,%
\footnote{%
We denote the $s$-variable derivative by $\displaystyle \dv{\sfq}{s} = \sfq'$, etc.
}
\begin{subequations}\label{eq:qp_lower_ex}
\begin{align}
    \sfq_1 & = \sfq' + u_0 \sfq 
    \\
    \sfq_2 & = \sfq'' + u_0 \sfq'  + u_1 \sfq - \sfp \sfq^2
    \\
    \sfq_3 & = \sfq''' + u_0 \sfq'' + u_1 \sfq' + u_2 \sfq - 4 \sfq \sfq' \sfp - \sfq^2 \sfp' - u_0 \sfq^2 \sfp  
    \, ,
    \label{eq:q3_lower}
    \\
    \sfp_1 & = \sfp' + v_0 \sfp
    \, ,
    \\
    \sfp_2 & = \sfp'' + v_0 \sfp' + v_1 \sfp - \sfp^2 \sfq
    \, ,
    \\ 
    \sfp_3 & = \sfp''' + v_0 \sfp'' + v_1 \sfp' + v_2 \sfp - 4 \sfp \sfp' \sfq - \sfp^2 \sfq' - v_0 \sfp^2 \sfq
    \, .
    \label{eq:p3_lower}
\end{align}
\end{subequations}

\subsection{Integral of motion}

We consider the integral of motion as discussed in Lemma~\ref{lem:IOM}.
In this case, we introduce the integral of motion for $\ell \le p$,
\begin{align}
    \mathsf{I}_\ell = u_\ell - (-1)^\ell v_\ell + \sum_{k=1}^\ell (-1)^k \qty( u_{\ell-k} v_{k-1} - \sfq_{\ell - k} \sfp_{k-1})
    \, .
\end{align}
Since all the functions are zero in the limit $s \to \infty$, the $s$-independent constant shall be zero,
\begin{align}
    \mathsf{I}_\ell = 0
    \, .
\end{align}
The lower degree cases are explicitly given as follows,
\begin{subequations}\label{eq:IOM_ex}
\begin{align}
    \mathsf{I}_0 & = u_0 - v_0
    \, ,
    \\
    \mathsf{I}_1 & = u_1 + v_1 - u_0 v_0 + \sfq \sfp 
    \, , \label{eq:IOM_ex1}
    \\
    \mathsf{I}_2 & = u_2 - v_2 + u_0 v_1 - u_1 v_0 - \sfq \sfp_1 + \sfq_1 \sfp
    \, .
\end{align}
\end{subequations}

\subsection{Nonlinear differential equations}\label{sec:nlin_eq}

We can derive closed differential equations for the wave functions $(\sfq,\sfp)$ by equating $(\sfq_p,\sfp_p)$ obtained as in \eqref{eq:qp_lower_ex} with another expressions~\eqref{eq:n_lower}.
For even $p = 2n$, the resulting differential equation for $\sfq(s) = \sfp(s)$ is known to be the $n$-th equation of the Painlev\'e II hierarchy~\cite{Claeys:2009CPAM,Akemann:2012bw,LeDoussal:2018dls}.
The lower degree equations are given by%
\footnote{%
We use the convention for the $s$-derivatives, $\sfq^{(1)}=\sfq'$, $\sfq^{(2)}=\sfq''$, etc.
}
\begin{subequations}
\begin{align}
    p = 2: \hspace{.9em} \sfq'' & = \sfq \qty(s + 2 \sfq^2) = s \sfq + 2\sfq^3 
    \, , \\
    p = 4: \ \sfq^{(4)} & = \sfq \qty(s - 6 \sfq^4 + 10 \sfq'^2 + 10 \sfq \sfq'')
    \, , \\
    p = 6: \ \sfq^{(6)} & = \sfq \left(s+ 20 \sfq^6-140 \sfq^2 \sfq'^2-70 \sfq^3 \sfq''+42
    \sfq''^2+56 \sfq' \sfq''' + 14 \sfq \sfq''''\right)
    +70 \sfq'^2 \sfq''
    \, .
\end{align}
\end{subequations}

For odd $p$, we instead obtain coupled nonlinear equations.
We first consider the simplest example of odd degree, $p = 3$.
In this case, from \eqref{eq:n_lower}, the wave functions for $k = 3$ are written in terms of $(\sfq_0,\sfp_0)=(\sfq,\sfp)$,
\begin{subequations}
\begin{align}
    \sfq_3 & = v_0 \sfq'' + (u_0 v_0 - v_1) \sfq' + (+ s + v_2 - u_0 v_1 + u_1 v_0) \sfq - v_0 \sfq^2 \sfp
    \, ,
    \\
    \sfp_3 & = u_0 \sfp'' + (u_0 v_0 - u_1) \sfp' + (- s + u_2 - u_1 v_0 + u_0 v_1) \sfp - u_0 \sfq \sfp^2
    \, .
\end{align}
\end{subequations}
Equating these expressions with~\eqref{eq:q3_lower} and \eqref{eq:p3_lower}, together with the integrals of motion~\eqref{eq:IOM_ex}, we obtain coupled nonlinear differential equations for the wave functions $(\sfq,\sfp)$,
\begin{subequations}
\begin{align}
    \sfq^{(3)} & = + s \sfq + 6 \sfp\sfq\sfq'
    \, ,
    \\
    \sfp^{(3)} & = - s \sfp + 6 \sfq\sfp\sfp'
    \, .
\end{align}
\end{subequations}
We may apply the same process to obtain differential equations for higher degree cases.
\begin{itemize}
    \item $p = 5$
\begin{subequations}
\begin{align} 
    \sfq^{(5)} & = \sfq \left(+s+10 \sfq' \sfp''+10 \sfp' \sfq''+10 \sfp \sfq'''\right) -30 \sfp^2 \sfq^2 \sfq' +10 \sfq' \left(\sfp' \sfq' +2 \sfp \sfq''\right)
    \, , \\
    \sfp^{(5)} & = \sfp \left(-s+10 \sfq' \sfp'' +10 \sfp' \sfq''+10 \sfq \sfp'''\right) -30 \sfp^2 \sfq^2 \sfp'+10 \sfp' \left(\sfp' \sfq'+2 \sfq \sfp''\right)
    \, .
\end{align}
\end{subequations}

\item $p = 7$
\begin{subequations}
\begin{align}
    \sfq^{(7)} & = \sfq \left(+s+42 \sfp'' \sfq^{(3)}+14 \left(-20 \sfp^2 \sfq' \sfq''+2 \sfq'' \sfp^{(3)}+\sfq' \sfp^{(4)}+2 \sfp' \sfq^{(4)}+\sfp \left(-20 \sfp' \sfq'^2+\sfq^{(5)}\right)\right)\right) 
    \nonumber \\
    & \quad
    + 140 \sfp^3 \sfq^3 \sfq'-70 \sfq^2 \left(\sfp'^2 \sfq'+2 \sfp \sfp' \sfq''+\sfp \left(2 \sfq' \sfp''+\sfp \sfq^{(3)}\right)\right)
    \nonumber \\
    & \quad
    +14 \left(-5 \sfp^2
\sfq'^3+5 \sfp' \sfq''^2+2 \sfq'^2 \sfp^{(3)}+\sfq' \left(8 \sfp'' \sfq''+7 \sfp' \sfq^{(3)}\right) + \sfp \left(5 \sfq'' \sfq^{(3)}+3 \sfq'   \sfq^{(4)}\right)\right)
    \, , \\
    \sfp^{(7)} & = \sfp
\left(-s+42 \sfq'' \sfp^{(3)}+14 \left(-20 \sfq^2 \sfp' \sfp''+2 \sfp'' \sfq^{(3)}+2
    \sfq' \sfp^{(4)}+\sfp' \sfq^{(4)}+\sfq \left(-20 \sfp'^2
    \sfq'+\sfp^{(5)}\right)\right)\right)
    \nonumber \\
    & \quad
    + 140 \sfp^3 \sfq^3 \sfp'-70 \sfp^2 \left(\sfp' \left(\sfq'^2+2 \sfq \sfq''\right)+\sfq
    \left(2 \sfq' \sfp''+\sfq \sfp^{(3)}\right)\right)
    \nonumber \\
    & \quad
    +14 \left(-5
\sfq^2 \sfp'^3+5 \sfq' \sfp''^2+\sfp' \left(8 \sfp'' \sfq''+7 \sfq' \sfp^{(3)}+2 \sfp'
    \sfq^{(3)}\right)+\sfq \left(5 \sfp'' \sfp^{(3)}+3 \sfp'
    \sfp^{(4)}\right)\right)
    \, .
\end{align}
\end{subequations}

\if0
\item $n = 9$
\begin{subequations}
\begin{align}
    q^{(9)} & = 
    q \left(x+6 \left(210 p^3 q'^3-434 p'^2 q' q''+19 q^{(3)} p^{(4)}+23
   p^{(3)} q^{(4)}-42 p^2 \left(5 q'' q^{(3)}+3
q' q^{(4)}\right)
    \right.\right.
    \nonumber \\
    & \quad
    +9 q'' p^{(5)}+19 p'' q^{(5)}+3 q'
   p^{(6)}+p' \left(-350 q'^2 p''+9 q^{(6)}\right)
    \nonumber \\
    & \quad
    \left.\left.
   +p \left(-28
\left(13 p' q''^2+6 q'^2 p^{(3)}+q' \left(22 p'' q''+19 p'
   q^{(3)}\right)\right)+3 q^{(7)}\right)\right)\right)
   \nonumber \\
    & \quad
    -630 p^4 q^4 q'-798 p'^2 q'^3+420 p q^3 \left(3 p'^2 q'+3
   p p' q''+p \left(3 q' p''+p q^{(3)}\right)\right)
   \nonumber \\
    & \quad
   +42
q^2 \left(60 p^3 q' q''-3 q' \left(3 p''^2+4 p'
    p^{(3)}\right)-11 p' \left(2 p'' q''+p' q^{(3)}\right)
    \right.
    \nonumber \\
    & \quad
    \left.
    -6 p
\left(2 q'' p^{(3)}+3 p'' q^{(3)}+q' p^{(4)}+2 p'
    q^{(4)}\right)+p^2 \left(90 p' q'^2-3 q^{(5)}\right)\right)
    \nonumber \\
    & \quad
    +6
p' \left(63 q^{(3)}^2+14 q'' \left(-44 p q'^2+7 q^{(4)}\right)+44
q' q^{(5)}\right)
    \nonumber \\
    & \quad
    +6 \left(63 q''^2 p^{(3)}-7 p^2
q' \left(31 q''^2+23 q' q^{(3)}\right)+2 q'' \left(77 p''
   q^{(3)}+27 q' p^{(4)}\right)
   \right.
   \nonumber \\
   & \quad
   \left.
   +q' \left(86 p^{(3)} q^{(3)}+82
p'' q^{(4)}+9 q' p^{(5)}\right)+2 p \left(-84 q'^3 p''+21
   q^{(3)} q^{(4)}+14 q'' q^{(5)}+6 q' q^{(6)} \right)\right)
\end{align}
\end{subequations}
\fi

\end{itemize}

We remark that these nonlinear differential equations have been recently obtained in the context of the Riemann--Hilbert problem associated with the Pearcey kernel~\cite{Chouteau:2023MPAG}.

\subsection{Asymptotic behavior and the boundary condition}

We consider the asymptotic behavior in the large $s$ limit.
Recalling that the resolvent kernel behaves as $R$, $L \to 0$ in the limit $s \to +\infty$ ($I \to \emptyset$), together with \eqref{eq:qp_RL}, we see that the auxiliary wave functions behave as 
\begin{align}
    \sfq(s) \ \xrightarrow{s \to \infty} \ \phi(s)
    \, , \qquad
    \sfp(s) \ \xrightarrow{s \to \infty} \ \psi(s)
    \, .
    \label{eq:qp_asympt}
\end{align}
We then obtain the asymptotic behavior of the bilinear form $\sfq^{(m)} \sfp^{(n)}$ from \eqref{eq:Ai_asymp1} and \eqref{eq:Ai_asymp2} as follows,
\begin{align}
    \sfq^{(m)} \sfp^{(n)} \ \xrightarrow{s \to \infty} \ O(s^{-1+\frac{1+m+n}{p}})
    \, ,
\end{align}
which goes to zero in the limit $s \to \infty$ if $n + m < p - 1$.
In fact, in the limit $s \to \infty$, all the nonlinear differential equations obtained in \S\ref{sec:nlin_eq} are asymptotically reduced to
\begin{align}
    \dv[p]{\sfq}{s} = s \sfq
    \, , \qquad
    \dv[p]{\sfp}{s} = (-1)^p s \sfp
    \, ,
    \label{eq:qp_asympt_ODE}
\end{align}
which is consistent with the asymptotic behavior \eqref{eq:qp_asympt}.
\begin{proposition}
The auxiliary wave functions $(\sfq,\sfp)$ obey the differential equation \eqref{eq:qp_asympt_ODE} in the limit $s \to \infty$, and hence they are asymptotic to the $p$-Airy function \eqref{eq:qp_asympt} for arbitrary $p \ge 2$.
\end{proposition}
It has been already established that this statement holds for even $p$.
One can specify the boundary condition of the nonlinear differential equations, which is an analog of the Hastings--McLeod solution to the Painlevé II equation asymptotic to the Airy function~\cite{Hastings:1980ARMA}.

\begin{proposition}
We have the large gap behavior of the Fredholm determinant,
\begin{align}
    F(s) \ \xrightarrow{s \to - \infty} \ \exp \qty( - \tilde{C}_p |s|^{\frac{2}{p}+2})
\end{align}
with a $p$-dependent positive constant $\tilde{C}_p$.
\end{proposition}
\begin{proof}
The proof is parallel with Proposition~\ref{prop:large_gap}.
In this case, we rewrite the $s$-derivative of the Hamiltonian~\eqref{eq:H_der_TW} in terms of $(u_k,v_k)_{k = 0,1}$ using the integral of motion~\eqref{eq:IOM_ex1}.
Then, we see that $H' \xrightarrow{s \to - \infty} O(s^{\frac{2}{p}})$ as in Lemma~\ref{lem:H_der_asymp}.
Integrating this twice, we obtain the large gap behavior of the Fredholm determinant.
\end{proof}

%% References %%
%\bibliographystyle{utphys}
%\bibliographystyle{amsalpha_mod}
\bibliographystyle{ytamsalpha}
\bibliography{ref}

\end{document}